\documentclass[letterpaper, 10 pt, conference]{ieeeconf}  % Comment this line out if you need a4paper

\IEEEoverridecommandlockouts                              % This command is only needed if 
\usepackage{amsmath} % assumes amsmath package installed
\usepackage{amssymb}  % assumes amsmath package installed
\usepackage{algorithm}
\usepackage{algorithmicx}
\usepackage{algpseudocode}
\usepackage{mathrsfs}
\usepackage{mathtools}
\let\labelindent\relax
\usepackage{enumitem}
\setlist[enumerate]{left=0pt, label=(\arabic*)}
\setlist[itemize]{left=0pt}
\usepackage{xcolor}

\newtheorem{theorem}{Theorem}
\newtheorem{corollary}{Corollary}
\newtheorem{lemma}{Lemma}
\newtheorem{remark}{Remark}
\newtheorem{definition}{Definition}
\newtheorem{assumption}{Assumption}
\newtheorem{proposition}{Proposition}
\makeatletter
\newcommand{\LongState}[1]{%
  \State \parbox[t]{\dimexpr\linewidth-\ALG@thistlm\relax}{#1}%
}
\makeatother

\title{\LARGE \bf
Non-asymptotic Bounds of Learning-based Linear MPC With Input Constraints and Unbounded Stochastic Noise
}

\author{Changyi Lei, Seth Siriya, Dragan Ne\v{s}i\'{c} and Ye Pu % <-this % stops a space
\thanks{*The work of Changyi Lei was supported by the Australian Government Research Training Program (RTP) Scholarship and the Australian Research Council under Project DE220101527.  
The work of Dragan Nešić was supported by the Australian Research Council through the Discovery Project under Grant DP250100300. 
The work of Ye Pu was supported by the Australian Research Council under Project DE220101527.}% <-this % stops a space
\thanks{The authors are with the Department of Electrical and Electronic Engineering, University of Melbourne, Parkville, VIC 3010, Australia 
(e-mail: chalei1@student.unimelb.edu.au; siriya@irt.uni-hannover.de; \{dnesic, ye.pu\}@unimelb.edu.au).}%
}

\begin{document}

\maketitle
\thispagestyle{empty}
\pagestyle{empty}

%%%%%%%%%%%%%%%%%%%%%%%%%%%%%%%%%%%%%%%%%%%%%%%%%%%%%%%%%%%%%%%%%%%%%%%%%%%%%%%%
\begin{abstract}

This paper studies learning-based model predictive control (MPC) for stabilizing unknown discrete-time linear systems with hard input constraints and additive unbounded sub-Gaussian disturbances. We adopt a certainty-equivalence (CE) design that combines a switching MPC control law with online regularized least-squares (RLS) parameter estimation. The resulting switching control law blends the MPC with a saturated deadbeat controller, ensuring global closed-loop stability. Building upon non-asymptotic error bound of least-squares, we derive non-asymptotic, high-probability stability bounds for the closed-loop system under the proposed switching controller. Numerical experiments illustrate and support the theoretical findings.

\end{abstract}

%%%%%%%%%%%%%%%%%%%%%%%%%%%%%%%%%%%%%%%%%%%%%%%%%%%%%%%%%%%%%%%%%%%%%%%%%%%%%%%%
\section{INTRODUCTION}
Model predictive control (MPC) is a widely used methodology for the control of multivariable systems with constraints. At each time instant, MPC solves a finite-horizon optimal control problem based on a model of the plant and applies the first portion of the optimal input sequence. This optimization-based framework allows one to systematically handle hard input and state constraints and optimize multi-step performance objectives \cite{borrelli2017predictive}. A fundamental limitation of MPC, however, is its strong dependence on an accurate model of the plant \cite{rawlings2017model}. Broadly speaking, two main approaches have emerged to address model uncertainty in MPC: robust MPC \cite{mayne2005robust} and learning-based MPC \cite{hewing2020learning}. This paper focuses on the latter perspective.

Learning-based MPC typically integrates data-driven model refinement with predictive control. Depending on how predictive information is obtained from data, existing methods can be broadly grouped into nonparametric model-based methods, direct data-driven methods, and parametric model-based methods \cite{hewing2020learning}. Nonparametric approaches construct flexible surrogate models, such as Gaussian-process, RKHS/kernel, or kinky-inference models, and then use the associated uncertainty quantification to support robust or cautious MPC design \cite{koller2018learning,hewing2019cautious,maddalena2021kpc,limon2017learning}. While these methods are attractive for general nonlinear systems, they often face scalability and uncertainty-propagation challenges over the MPC horizon \cite{scampicchio2025gaussian}. Direct data-driven MPC  \cite{coulson2019data,berberich2020data} bypasses explicit model identification by using measured input-output trajectories directly for prediction and have been shown to recover model-based MPC behavior for unknown LTI systems with many variants \cite{berberich2021terminal,bongard2022robust,kloppelt2025novel}. However, this method is sensitive to the noise level in the data \cite{sassella2022noise}. Therefore, for the rest of the paper, we focus on the parametric model-based method.

The third category consists of parametric model-based learning for MPC design. These methods assume a known model structure and use online data either to estimate unknown parameters or to refine a set of admissible models. Early adaptive MPC schemes for constrained linear systems, such as \cite{di2016indirect}, combined parameter estimation with robust MPC ingredients with terminal constraints to obtain recursive feasibility, constraint satisfaction, and stability. Set-membership approaches further update a polytopic uncertainty set from input-output data and enforce constraints robustly over the updated set \cite{tanaskovic2014adaptive}. Building on this viewpoint, later works incorporated stabilizing terminal ingredients, homothetic tubes, recursive least-squares updates, and robust adaptive MPC designs to establish recursive feasibility, robust constraint satisfaction, and closed-loop stability \cite{kohler2019linear,lorenzen2019robust,zhang2020adaptive}. Extensions to nonlinear discrete-time and continuous-time systems were developed in \cite{kohler2021robust,kohler2026certainty,sasfi2023robust}, while \cite{parsi2023dual} introduced a dual adaptive MPC formulation that accounts for future set-membership updates inside the control optimization. A related direction is Koopman-based MPC, which learns a finite-dimensional lifted predictor of nonlinear dynamics. Existing results then apply linear, bilinear, or robust MPC in the lifted space \cite{zhang2022robust,worthmann2024data}. However, the finite-dimensional Koopman projection error remains the main limiting factor \cite{schimperna2025stability}.

Despite the method-specific disadvantages discussed above, existing parametric model-based MPC approaches share several common restrictions. First, they typically assume an a priori known compact parameter set containing the true system parameter. While this is convenient for set-membership estimation and robust constraint tightening, it is restrictive in practice, especially for open-loop unstable systems, where classical identification results often rely on stability, mixing, or non-explosiveness assumptions \cite{ljung1995system,yu1994rates}. Second, robust adaptive MPC schemes usually require stabilizing ingredients before online learning begins, e.g. \cite{lorenzen2019robust}. Thus, stability is partly presupposed rather than achieved through learning and control. Moreover, terminal constraints are difficult to construct, particularly for nonlinear systems, and typically yield only regional stability guarantees because the required invariance and reachability conditions hold on restricted state sets \cite{limon2005enlarging,rawlings2017model}. Consequently, existing parametric learning-based MPC results do not fully resolve stability under genuine model uncertainty. Although later works such as \cite{bold2024data} avoids terminal constraints in an EDMD-based MPC setting, its guarantee remains semiglobal and depends on region-dependent constants that are difficult to verify explicitly.

The above limitations are closely tied to how model uncertainty is handled in the stability analysis of learning-based MPC. With a pre-stabilizing controller, a compact uncertainty set, and bounded noise, learning can be confined to a prescribed region. In the setting considered here, however, learning and stabilization must occur simultaneously, and the transient state may become arbitrarily large due to instability and unbounded disturbances. Stability therefore cannot rely on compactness, terminal ingredients, or regional Lyapunov arguments. Instead, it requires 1) finite-time identification bounds valid along unstable trajectories, and 2) a global analysis that propagates the resulting model error through the MPC value function and feedback law.

The first requirement connects the present problem to non-asymptotic system identification. In contrast to classical consistency theory, which describes limiting estimator behavior \cite{ljung1995system}, non-asymptotic results provide finite-time, high-probability error bounds of the form $|\hat{\theta}_T-\theta^\star| \leq \epsilon_T$. Such bounds have been developed for linear dynamical systems from single-trajectory data \cite{simchowitz2018learning}, including unstable systems \cite{faradonbeh2018finite,sarkar2019near,siriya2024non}, partially observed systems \cite{oymak2019non,zheng2020non}, and control-oriented settings \cite{wagenmaker2020active,dean2020sample}. These results are essential here because they allow model uncertainty to be quantified during learning, even before closed-loop stability has been established. In the context of adaptive and learning-based control, these bounds allow the mismatch between the true and estimated dynamics to be treated as a finite-time, vanishing perturbation in the Lyapunov drift, thereby enabling high-probability stability guarantees during simultaneous learning and control \cite{siriya2025framework}.

Motivated by the above discussions, this paper studies learning-based control for unknown linear systems subject to hard input saturation and additive unbounded sub-Gaussian disturbances. The objective is to achieve simultaneous online system identification and global stochastic stabilization using MPC, without assuming a pre-stabilizing controller, a known compact uncertainty set for robust MPC design, or stabilizing terminal ingredients. The main contributions are summarized as follows.

\textbf{(1)} We propose a certainty-equivalent adaptive MPC scheme for unknown linear systems with hard input saturation and unbounded stochastic disturbances. The controller combines recursive least-squares estimation with an excitation-based MPC input while enforcing $|u_t|\le u_{\max}$ at all times. Unlike standard adaptive MPC designs, it does not require a pre-stabilizing controller, a known compact uncertainty set, or terminal constraints/terminal costs. The resulting stability guarantee is global with respect to the initial condition.

% \textbf{(1)} We propose a certainty-equivalent adaptive MPC scheme for unknown linear systems with hard input saturation and unbounded stochastic disturbances. The controller combines recursive least-squares parameter estimation with an excited control input computed from the current parameter estimate, while enforcing the input constraint $\|u_t\|\le u_{\max}$ at every time step. In contrast to many existing adaptive and learning-based MPC schemes, the proposed design does not rely on a pre-stabilizing controller, a known compact uncertainty set for robust constraint tightening, or terminal constraints/terminal costs to certify stability. The resulting closed-loop guarantee is global in the sense that the initial condition is arbitrary.

\textbf{(2)} To make the certainty-equivalent design amenable to global stability analysis, we introduce a hysteresis switching mechanism \cite{nesic2004framework} between two controllers. This mechanism avoids abrupt switching near the boundary and yields global Lipschitz continuity of the closed-loop control law with respect to the estimated parameters, which is essential for propagating finite-time identification errors through the control input and the MPC value function.

% : an MPC law used inside a compact inner region $\mathcal{X}_{\mathrm{in}}$ and a radially saturated deadbeat-type controller used outside a larger compact region $\mathcal{X}_{\mathrm{out}}\supset \mathcal{X}_{\mathrm{in}}$. The annular region $\mathcal{X}_{\mathrm{out}}\setminus \mathcal{X}_{\mathrm{in}}$ serves as a hysteresis band in which the active controller is held fixed. 

\textbf{(3)} We establish non-asymptotic high-probability stability bounds for the resulting learning-based closed loop. The analysis combines finite-time RLS identification bounds \cite{siriya2024non} with a global stochastic stability argument based on MPC without terminal constraints \cite{grune2012nmpc} and perturbation analysis. This yields explicit high-probability bounds characterizing the closed-loop behavior under simultaneous learning and control.

\textit{Notation}: Let $\mathbb{N}_0:=\mathbb{N}\cup{0}$ and let $\mathbb{R}_{\geq0}$ denote the set of nonnegative real numbers. For a matrix $A$, let $|A|$, $|A|_F$, $|A|_\infty$, $A^\intercal$, and $A^\dagger$ denote its spectral norm, Frobenius norm, infinity norm, transpose, and Moore--Penrose inverse, respectively. For a vector $x$, $|x|$ denotes its Euclidean norm. For positive definite $A$, $\lambda_{\max}(A)$ denotes its largest eigenvalue. Let $\mathcal{S}^{d-1}$ denote the unit sphere in $\mathbb{R}^d$, and let $\mathrm{Id}$ denote the identity map. We use standard definitions of class $\mathcal{K}$, $\mathcal{K}_\infty$, and $\mathcal{L}$ functions. For $r>0$, define
\(
\operatorname{sat}_r(x):=
\begin{cases}
x, & |x|\le r,\\
rx/|x|, & |x|>r.
\end{cases}
\)
We write $X\overset{d}{\sim}Y$ for equality in distribution. Given $\bar A\in\mathbb{R}^{n\times n}$ and $\bar B\in\mathbb{R}^{n\times m}$, define
\(
\mathcal{R}_{\kappa}(\bar A,\bar B)
:=
\begin{bmatrix}
\bar B & \cdots & \bar A^{\kappa-1}\bar B
\end{bmatrix}
\in\mathbb{R}^{n\times \kappa m}.
\)

\section{Problem Setup}
We consider the following stochastic discrete-time linear system
\begin{equation}
    x_{t+1} = Ax_t + Bu_t + w_t \label{sys1}
\end{equation}
for $t \in \mathbb{N}_0$ where $x_t \in \mathbb{R}^n$ is the state, $u_t \in \mathbb{U} \subseteq \mathbb{R}^m$ is the control input, $w_t \in \mathbb{W} \subseteq \mathbb{R}^n$ is the random disturbance, and $x_0 \in \mathbb{R}^n$ is the initial state. The set $\mathbb{U}= \{u \big| \lVert u \rVert_\infty \leq u_{\max} \} $ is the hard input constraint set with a constant upper bound $u_{\max}>0$. Denote $\theta^* = [A \ B]$ as the true but unknown parameter of the system where $A \in \mathbb{R}^{n \times n}$ and $B \in \mathbb{R}^{n \times m}$. We also introduce $\hat{\theta}_t = [\hat{A}_t \ \hat{B}_t]$ that denotes the estimated parameter and $\tilde{\theta}_t = \hat{\theta}_t - \theta^*$ denotes the estimation error, both at time $t$. The following assumptions are imposed on system (\ref{sys1}) with some justification in Remark \ref{remark:1}.

\begin{assumption} (Disturbance)
    The (potentially unbounded) disturbance sequence $(w_t)_{t \in \mathbb{N}_0}$ is an 0-mean independent and identically distributed (i.i.d.) sequence across time $t$, and satisfies $\sigma_W^2$-sub-Gaussian distribution \cite[Definition 3.4.1]{VershyninHDP2}, with a diagonal covariance matrix $\Sigma_W$. Further, there exists a known lower bound $\underline{\sigma}_W^2>0$ on the minimal eigenvalue of $\Sigma_W$.
    \label{assump1}
\end{assumption}

\begin{assumption} (System Matrices)
    Matrix $A$ satisfies $\lVert A \rVert \leq 1$ and $(A,B)$ is $\kappa$-step reachable, namely $\text{rank}(\mathcal{R}_{\kappa}(A,B))=n$ for some $\kappa \in \mathbb{N}$ and $\sigma_{min}(\mathcal{R}_{\kappa}(A,B)) \geq \underline{\sigma}_B > 0$ where $\underline{\sigma}_B>0$ is a known constant. For simplicity, we denote $\mathcal{R}_*:=\mathcal{R}_{\kappa}(A,B)$.
    \label{assump2}
\end{assumption}

% \begin{remark}
% \label{remark:1}
%     In Assumption \ref{assump1}, we allow the process noise to admit a sub-Gaussian distribution that is potentially unbounded, which encompasses a wide range of stochastic noise. Moreover, the lower bound of the variance is assumed known for the proof of persistency of excitation as in \cite{li2023non}. This lower bound assumption is for simpler exposition of our results and does not restrict the generality of the system class considered. In particular, when the system is $\kappa$-step reachable as stated in Assumption~\ref{assump1}, sufficient excitation over every $\kappa$ steps can alternatively be ensured through suitable input perturbations (e.g. random injected noise). 

%     Assumption~\ref{assump2} ensures global stabilizability of the system under input saturation by $\|A\|\leq 1$, which is required since $w_k$ can take arbitrarily large values and may drive the state to any region of $\mathbb R^n$.
% \end{remark}

\begin{remark}
\label{remark:1} 
Assumption~\ref{assump1} is mild and does not significantly restrict the class of systems considered. It allows possibly unbounded sub-Gaussian disturbances, covering a broad class of stochastic noise. The known lower bound on the noise variance is used to establish persistency of excitation, as in \cite{li2023non}. This lower-bound assumption is adopted mainly for simpler exposition and can be relaxed. For example, when the system is \(\kappa\)-step reachable, sufficient excitation over every \(\kappa\) steps can alternatively be enforced through suitable input perturbations, such as injected random noise. In Assumption~\ref{assump2}, the condition $\|A\| \leq 1$ is necessary to guarantee global stabilizability under input saturation \cite{chatterjee2012mean}. The lower bound related to the reachability matrix is sufficient for guaranteeing the existence of control policies with input saturation that render the system mean-square bounded \cite{chatterjee2012mean}, which gives us hope that our learning-based MPC policy can render the system stochastic stability.
\end{remark}

\textbf{Control Objective: }Our goal is to design a learning-based MPC controller for system (\ref{sys1}) with unknown parameters $\theta^* = [A \ B]$ and additive unbounded stochastic noise $w_t$, to achieve simultaneous stabilization (i.e. control) and system identification (i.e. learning). Specifically, the system states should converge to a neighborhood proportional to the size of the noise around the origin under user-defined hard input constraints, alongside with the system parameters being estimated.

\section{Learning-based MPC-DB Switching Control Law and Theoretical Guarantees}
In this section, the proposed control algorithm and the main theorem of closed-loop stability are provided. The proofs of all theoretical results are deferred to the next section. We first define the $\kappa$-step sub-sampled system, and the sub-sampled state sequence $(\bar{x}_\tau)_{\tau \in \mathbb{N}_0}$ satisfying $\bar{x}_\tau = x_{\tau \kappa}$ evolves as
\begin{equation}
    \bar{x}_{\tau+1}=A^\kappa \bar{x}_\tau + \mathcal{R}_* \bar{u}_\tau + \mathcal{R}_\kappa(A,I) \bar{w}_\tau,
    \label{sys1_subsampled}
\end{equation}
where
\[\bar{u}_\tau = \begin{bmatrix}
        u_{\tau \kappa} \\
        u_{\tau \kappa+1} \\
        \vdots \\
        u_{(\tau+1) \kappa-1}
    \end{bmatrix},
\bar{w}_\tau = \begin{bmatrix}
        w_{\tau \kappa} \\
        w_{\tau \kappa+1} \\
        \vdots \\
        w_{(\tau+1) \kappa-1}
    \end{bmatrix}.
\]
The controller design will be based on this sub-sampled system, which is 1-step reachable.

Overall, the MPC-DB switching control law has four components: 
\begin{itemize}
    \item a least square estimator $\hat{\theta}=[\hat{A} \ \hat{B}] \in \mathbb{R}^{n\times (n+m)}$ in \eqref{RLS}
    \item an additive excitation $v$ with magnitude $C \in (0,u_{\max})$ in \eqref{controlall}
    \item a saturated deadbeat (DB) controller $\pi_{\mathrm{DB}}(x,\hat{\theta},C)$ in \eqref{controllaw_sat}
    \item an MPC controller $\pi_{\mathrm{MPC}}(x,\hat{\theta})$ by solving \eqref{MPCProblem}
\end{itemize}
where the DB and MPC are combined according to a switching logic. At each timestep, the Regularized Least Square (RLS) algorithm is used to estimate the parameters. The estimated parameters will be fed into the switching controller for control law calculation. 

\subsection{Regularized Least Square}
Let $\hat{\theta}_t=[\hat{A}_t \ \hat{B}_t]$ denote the estimated parameter at time $t$.The RLS estimate is then derived as
\begin{equation}
    \hat{\theta}_t = \arg \min_{\theta \in \mathbb{R}^{n \times (n+m)}} \sum_{s=1}^t \left| x_s - \theta \begin{bmatrix}
        x_{s-1} \\ u_{s-1}
    \end{bmatrix}   \right|^2 + \psi \lVert \theta \rVert^2_F
    \label{RLS}
\end{equation}
where $\psi>0$ is the coefficient for the regularization term. Note that the initial estimates $\hat{A}_0,\hat{B}_0$ are chosen arbitrarily.  

\subsection{MPC-DB Switching Control Law}
For user-defined excitation level $C \in (0,u_{\max})$, the saturated deadbeat controller has the form 
\begin{equation}
    \pi_{\mathrm{DB}}(x,\hat{\theta},C) =\text{sat}_{u_{\max}-C}(-\mathcal{R}_\kappa(\hat{A}, \hat{B})^\dagger \hat{A}^\kappa x).
     \label{controllaw_sat}
\end{equation}
The MPC solution $\pi_{\mathrm{MPC}}(x,\hat{\theta}) \in \mathbb{R}^{\kappa m}$ that satisfies $ \lVert \pi_{\mathrm{MPC}}(x_t,\hat{\theta}) \rVert_\infty \leq u_{\max}-C$  is obtained by solving the problem (\ref{MPCProblem}) using the parameters $\hat{\theta}$. We employ stage costs in exponential form for technical reasons that will become clear in the proof of Lemma \ref{lem:lyasat}, with some explanation in Remark \ref{remark:expcosts}. Denote $\sigma(x)=\left(e^{|x|} - 1\right)^2$. Let the stage cost be defined as $\ell(x,u) := Q\sigma(x) + u^\intercal R u$, and the terminal cost as
$\ell_f(x) := Q\sigma(x)$, where $Q \in \mathbb{R}_{>0}$ and $R \in \mathbb{R}^{m \times m},\ R \succ 0$ are user-defined coefficients that adjust the relative weights of the costs. Let $N \in \mathbb{N}$ be the prediction horizon. 

\begin{flushleft}
\textbf{MPC Problem \(\mathbf{P}_N(x,\hat{A}, \hat{B})\)}:
\end{flushleft}
\begin{equation}
\label{MPCProblem}
\begin{aligned}
Y_{\mathrm{MPC}}(\hat{\theta},x)=&\min_{\{\hat{u}_i\}_{i=0}^{N-1}}  \ \sum_{i=0}^{N-1} \ell(\hat{x}_i, \hat{u}_i) + \ell_f(\hat{x}_N), \\
\text{s.t.} \ & \hat{x}_{0} = x, \\
& \hat{x}_{i+1}=\hat{A}_{\kappa \tau}^\kappa \hat{x}_i + \mathcal{R}_\kappa(\hat{A}_{\kappa \tau}, \hat{B}_{\kappa \tau}) \hat{u}_i, \\
& \hspace{6em}i \in \{0,\dots,N-1\} \\
& \|\hat{u}_i\|_\infty \leq u_{\max} - C.
\end{aligned}
\end{equation}

% \begin{remark}(Feasibility, exponential costs and state-independent prediction horizon) \label{remark:expcosts}
% Different from many existing works, we do not adopt terminal constraints to ensure stability to avoid the feasibility restrictions due to input constraints and to obtain global guarantees for all initial states. To compute a horizon length \(N\) that works uniformly over the whole domain, we employ an exponential stage cost. This choice ensures that the optimal finite-horizon cost is bounded by a constant multiple of the current stage cost. We use this property to derive a state-independent lower bound on \(N\) that guarantees closed-loop stability for all \(x\in\mathbb{R}^n\) (see Lemma \ref{MPCwithoutterminal}).
% \end{remark}
\begin{remark}(Feasibility, exponential costs, and state-independent prediction horizon)
\label{remark:expcosts}
To obtain global stability guarantees, we do not impose terminal constraints. Unlike many existing works, this avoids the feasibility restrictions induced by input constraints and prevents the stability guarantee from being limited to local.

To derive a prediction horizon \(N\) that is long enough to guarantee stability without terminal constraints, we use an exponential stage cost. This choice ensures that the optimal finite-horizon cost is bounded by a constant multiple of the current stage cost. Based on this property, we derive a state-independent constant lower bound on \(N\) that guarantees closed-loop stability (see Lemma~\ref{MPCwithoutterminal}).
\end{remark}

We denote $Y_{\mathrm{MPC}}(\hat{\theta}_{\kappa \tau},x)$ as a certainty-equivalence (CE) value function for the closed-loop system under policy $\pi_{\mathrm{MPC}}(\cdot,\hat{\theta}_{\kappa \tau})$, which is the minimal value attained by solving problem \(\mathbf{P}_N(x,\hat{A}_{\kappa \tau}, \hat{B}_{\kappa \tau})\).
% \begin{equation}
% \begin{aligned}
% Y_{\mathrm{MPC}}(\hat{\theta}_{\kappa \tau},x) \triangleq & \min_{\{\pi_i\}_{i=0}^{N-1}}
% \left[ \sum_{i=0}^{N-1} \ell(\hat{x}_i, \pi_i) + \ell_f(\hat{x}_N) \right], \\
% \text{s.t.} \quad 
% & \hat{x}_0 = x, \\
% & \hat{x}_{i+1} = \hat{A}_t \hat{x}_i + \hat{B}_t \pi_i, \quad i = 0, \dots, N-1, \\
% & \|\pi_i\|_\infty \le u_{\max}-C, \quad i = 0, \dots, N-1.
% \end{aligned} \label{eq:Y1hat}
% \end{equation}
Similarly, $Y_{\mathrm{DB}}(\hat{\theta}_{\kappa \tau},x)$ is defined as a CE realized cost for the closed-loop system under $\pi_{\mathrm{DB}}(\cdot,\hat{\theta}_{\kappa \tau},C)$:
\begin{equation}
\begin{split}
    &Y_{\mathrm{DB}}(\hat{\theta}_{\kappa \tau},x) = \sum_{i=0}^{N-1}\ell(\hat{x}_i,\pi_{\mathrm{DB}}(\hat{x}_i,\hat{\theta}_{\kappa \tau},C)) + \ell_f(\hat{x}_N),\\
&\text{s.t.} \  \hat{x}_0 = x, \\
&\quad \ \ \hat{x}_{i+1}=\hat{A}_{\kappa \tau}^\kappa \hat{x}_i + \mathcal{R}_\kappa(\hat{A}_{\kappa \tau}, \hat{B}_{\kappa \tau}) \pi_{\mathrm{DB}}(\hat{x}_i,\hat{\theta}_{\kappa \tau},C),\\
&  \hspace{12em}i \in \{0,\dots,N-1\}.
\end{split} \label{eq:Y2hat}
\end{equation}

We will now show the MPC-DB switching control law with hysteresis switching. Let the mode variable be \(s_\tau \in \{\text{MPC},\text{DB}\}\) and denote the next mode by \(s_{\tau+1}\).
Fix user-chosen constants \(c_{\mathrm{MPC}}>0\), \(c_{\mathrm{DB}}>0\), and a hysteresis width \(c_3>0\) such that
\begin{equation}
c_{\mathrm{MPC}} - c_3 \;\ge\; c_{\mathrm{DB}} .
\label{eq:hyswitchlevelsetconstants}
\end{equation}
For every $\kappa$ steps, the mode update is
\begin{equation}
s_{\tau+1} =
\begin{cases}
\text{MPC} , & \text{if } Y_{\mathrm{DB}}(\hat{\theta}_{\kappa \tau},\bar{x}_\tau) < c_{\mathrm{DB}} \\
\text{DB}, & \text{if } Y_{\mathrm{MPC}}(\hat{\theta}_{\kappa \tau},\bar{x}_\tau) \ge c_{\mathrm{MPC}}\\
s_\tau, & \text{otherwise} \\
\end{cases}
\label{eq:switch}
\end{equation}
where the third line implements the hysteresis action. Given an estimate \(\hat\theta_{\kappa \tau}\), the CE control law is
\begin{equation}
\pi_{s_\tau}(\bar{x}_\tau,\hat\theta_{\kappa \tau}) =
\begin{cases}
\pi_{\mathrm{MPC}}(\bar{x}_\tau,\hat\theta_{\kappa \tau})& \text{if } s_\tau=\text{MPC}\\
\pi_{\mathrm{DB}}(\bar{x}_\tau,\hat\theta_{\kappa \tau},C)& \text{if } s_\tau=\text{DB}
\end{cases}
\label{eq:nominal-controller}
\end{equation}
and the applied input injects exploration noise:
\begin{equation}
\begin{split}
    \bar{u}_\tau &= \alpha(\bar{x}_\tau,\bar{v}_\tau,\hat\theta_{\kappa \tau}) = \pi_{s_{\tau+1}}(\bar{x}_\tau,\hat\theta_{\kappa \tau}) + \bar{v}_\tau,\\
\bar{v}_\tau &\overset{\mathrm{i.i.d.}}{\sim} \mathrm{Uni}([-C,C]^{\kappa m}).
\end{split}
\label{controlall}
\end{equation}
Equivalently, \(\bar{v}_\tau \in \mathbb{V}:=\{v\in\mathbb{R}^{\kappa m}:\|v\|_\infty\le C\}\) with a uniform density over the hypercube. The hard input constraints are satisfied as $\lVert u_t \rVert_\infty \leq \lVert \bar{u}_\tau \rVert_\infty \leq \lVert \pi_s(\cdot) \rVert_\infty+\lVert \bar{v}_\tau \rVert_\infty \leq u_{\max}$. For implementation, we now map the augmented control input $\bar{u}_\tau\in \mathbb{R}^{\kappa m}$ in \eqref{controlall} to each $u_t \in \mathbb{R}^m$ of system \eqref{sys1}. Define the selection matrix $S_i =
\begin{bmatrix}
0_{m\times (i-1)m} & I_m & 0_{m\times (\kappa-i)m}
\end{bmatrix},$ which picks out the $i^{th}$ subsequence of size m. At each timestep, we use the matrix $S_i$ to select the corresponding input $u_t$ of size $m$ from $\bar{u}_\tau$ and apply to the system in sequence:
\begin{equation}
    u_{\tau \kappa + i}=S_i\pi_{s_{\tau+1}}(\bar{x}_\tau,\hat\theta_{\kappa \tau}) +S_i \bar{v}_\tau, \, i=1,2,...,\kappa.
    \label{eq:controlpolicyperstep}
\end{equation}
\begin{remark}(Hysteresis switching for global Lipschitz continuity)
\label{remark:hys_lip_con}
The hysteresis switching between the two controllers is introduced to ensure that the resulting control law is globally Lipschitz continuous with respect to the system parameters for all $x \in \mathbb{R}^n$. This property is important for the stability analysis, since it prevents small parametric uncertainty from causing arbitrarily large changes in the closed-loop stability. This can be illustrated by a simple counterexample showing that the MPC law need not be globally Lipschitz continuous in the parameters. Consider
\(
x_{t+1}=
\begin{bmatrix}
1 & \epsilon\\
0 & 1
\end{bmatrix}x_t+u_t,
\)
with finite-horizon cost
\(
J(x_t,u_t)=|x_t|^2+|u_t|^2+|x_{t+1}|^2.
\)
When \(\epsilon=0\), the optimal input is
\(
u_t^*=-\frac{1}{2}x_t.
\)
When \(\epsilon>0\), the optimal input becomes
\(
\bar u_t^*=
-\frac{1}{2}
\begin{bmatrix}
x_1+\epsilon x_2\\
x_2
\end{bmatrix}.
\)
Hence,
\(
\frac{\partial \bar u_t^*}{\partial \epsilon}
=
-\frac{1}{2}
\begin{bmatrix}
x_2\\
0
\end{bmatrix},
\)
whose norm is proportional to \(|x_2|\). Therefore, the MPC controller is Lipschitz continuous in \(\epsilon\) only on bounded subsets of the state space.
\end{remark}

With all the components described, Alg.\ref{algoall} summarizes our proposed control strategy. 
\begin{algorithm}
\caption{Learning-based MPC-DB Switching Control}
\begin{algorithmic}[1]
\State \textbf{Input:} $u_{\max}>0, C \in (0,u_{\max}), \hat{\theta}_0 \in \mathbb{R}^{n \times (n+m)}, \kappa \in \mathbb{N}$
\State \textbf{Select:} $s_0=\text{DB},N \in \mathbb{N},Q\in \mathbb{R}_{>0},R\in \mathbb{R}^{\kappa m\times \kappa m}$, and $c_{\mathrm{MPC}}, c_{\mathrm{DB}},c_3$ satisfying \eqref{eq:hyswitchlevelsetconstants}

\For{$\tau=0,1,...$}
    \State Measure $\bar{x}_\tau$
    \LongState{Calculate $Y_{\mathrm{MPC}}(\hat{\theta}_{\kappa\tau},\bar{x}_\tau)$ and
    $Y_{\mathrm{DB}}(\hat{\theta}_{\kappa\tau},\bar{x}_\tau)$ by solving problem
    $\mathbf{P}_N(\bar{x}_\tau,\hat{A}_{\kappa\tau},\hat{B}_{\kappa\tau})$ and \eqref{eq:Y2hat}}
    \State Determine $\pi_{s_{\tau+1}}(\bar{x}_\tau,\hat\theta_{\kappa \tau})$ using (\ref{eq:nominal-controller})
    \State Sample $\bar{v}_\tau$ i.i.d from $\mathrm{Uni}([-C,C]^{\kappa m})$

    \For{$i=0,1,..., \kappa-1$}
        \LongState{Apply the $\kappa$-step control law in sequence
        $u_{\tau\kappa+i}=S_i\pi_{s_{\tau+1}}(\bar{x}_\tau,\hat{\theta}_{\kappa\tau})+S_i\bar{v}_\tau$}
        \State Measure $x_{\tau \kappa + i+1}$ from (\ref{sys1})
        \State Compute $\hat{\theta}_{\tau \kappa + i+1}$ following (\ref{RLS})
    \EndFor
\EndFor
\end{algorithmic} \label{algoall}
\end{algorithm}

\subsection{Theoretical Results}
We are now ready to state our main theorem about the stability bound for the system (\ref{sys1}) under our proposed control strategy.
\begin{theorem}
Consider system~\eqref{sys1} under Assumptions~\ref{assump1}--\ref{assump2}, and the control laws in Algorithm~\ref{algoall} with switching thresholds $c_{\mathrm{MPC}},\,c_{\mathrm{DB}},\,c_3$ satisfying~\eqref{eq:hyswitchlevelsetconstants}. Given parameters $Q>0,R \succ 0$, confidence level $\delta \in (0,\frac{1}{2})$, initial state $x_0 \in \mathbb{R}^n$, and threshold $\frac{1}{2}c_3 < c_4 < c_{\mathrm{MPC}}$, suppose that $N\ge 2$ and satisfies 
\begin{equation}
q_3^{\max}\left( Q \gamma -Q \right)m_2^2 - q_3^{\min}Q + \frac{q_3^{\max} Q \gamma^2 }{N-1} <0  
\label{condition_ISS}
\end{equation} 
where
\begin{align}
\bar{E}_w 
&= \exp\Bigl(
\Big|
\ln\left( \mathbb{E} e^{|\mathcal{R}_* \bar{v}_\tau|}\right)
+ \ln\left( \mathbb{E} e^{| \mathcal{R}_\kappa(A,I) \bar{w}_\tau|} \right)
\Big|
\Bigr) - 1
\label{eq:Ew}
\\
m_2^2 
&= \tfrac{(Q + \lambda_{\max}(R)\| \mathcal{R}_*^\dagger A^\kappa \|^2)^2}
{Q^2 (N-1) (1-p^2)^2}
+ \tfrac{(Q + \lambda_{\max}(R)\| \mathcal{R}_*^\dagger A^\kappa \|^2)}
{Q (1-p^2)}
- 1 \ge 0,
\label{eq:m2}
\\
\gamma 
&= \tfrac{(Q + \lambda_{\max}(R)\| \mathcal{R}_*^\dagger A^\kappa \|^2)}
{Q(1-p^2)},
\label{eq:gamma}
\\
p 
&= \exp\bigl(-\sigma_{\min}(\mathcal{R}_*)(u_{\max}-C)\bigr), \\
q_2 &= \max\left\{\frac{4\, \gamma \, c_{\mathrm{MPC}}}{c_3}, \frac{1}{2 \gamma}, 2\gamma   \right\},\\
        q_3^{\min}&=\min \left\{1+\frac{1}{4\gamma}, q_2, \frac{1}{2\gamma}, 1 \right\}, \\
        q_3^{\max}& =\max \left\{1+\frac{1}{4\gamma}, q_2, \frac{1}{2\gamma}, 1 \right\}.
\label{eq:p}
\end{align}
 Then there exists a function $\tilde{\eta} \in \mathcal{L}$ and a constant $\tilde{c} \geq 0$ such that
\begin{equation}
\label{eq:per_step_state_bound}
    \mathbb{P}\bigl(\,|x_t| \le \tilde{\eta}(t) + \tilde{c}\,\bigr) \ge 1 - 3\delta, \forall t \ge  t_1
\end{equation}
where
\(
t_1 = t_1\!\bigl(A, B, Q, R, N, c_{\mathrm{MPC}}, c_{\mathrm{DB}}, c_3, c_4, u_{\max}, C\bigr).
\)
\label{stabilitytheorem}
\end{theorem}

\begin{remark}
Condition~\eqref{condition_ISS} balances (i) the noise growth via $\bar{E}_w$, and (ii) the receding-horizon term $\frac{1}{Q(N-1)(1-p^2)^2}$. The conclusion asserts a high-probability ultimate boundedness: after a finite transient $t_1$, the states remain within a decaying tube of radius $\tilde{\eta}(t)$ around a fixed offset $\tilde{c}$ with probability at least $1-\delta$. 
\end{remark}

For more intuitive understanding and practical implementation, the following corollary follows from Theorem \ref{stabilitytheorem}.
\begin{corollary}
Assume that $N\ge 2$ and Assumptions~\ref{assump1}--\ref{assump2} are satisfied. 
Using the same notation as Theorem~\ref{stabilitytheorem}, assign
\(
    c_{\mathrm{MPC}}=2c_3 .
\)
Select any $\xi$ such that
\(
    0<\xi<
    \min\left\{
    \frac{1}{(1+\bar{E}_w)^2},
    \frac{\sqrt{2}-1}{2}
    \right\}.
\)
Suppose that
\begin{equation}
\label{eq:umax-sufficient-new}
u_{\max}
\ge
C
+
\frac{1}{2\sigma_{\min}(\mathcal{R}_*)}
\ln\left(
\frac{2(1+\xi)}{\xi}
\right),
\end{equation}
\begin{equation}
\label{eq:R-sufficient-new}
\lambda_{\max}(R)
\le
\frac{\xi Q}
{2\|\mathcal{R}_*^\dagger A^\kappa\|^2},
\end{equation}
\begin{equation}
\label{eq:N-lb-new}
N>
\max\left\{
2,\,
1+
\frac{(1+\xi)^3}
{\frac{1}{16(1+\xi)^2}-\xi^2}
\right\}.
\end{equation}
Then condition~\eqref{condition_ISS} holds.

\label{corollarysufficientcondition}
\end{corollary}

% \begin{corollary} Assume that $N\ge 2$ and Assumptions~\ref{assump1}--\ref{assump2} are satisfied. Using the same notation as Theorem \ref{stabilitytheorem}. Fix any $\xi$ with $0<\xi<\frac{1}{(1+\bar{E}_w)^2}$ and assign $c_{\text{MPC}}=2c_3$. Suppose the following hold:
% \begin{equation}\label{eq:umax-sufficient}
% u_{\max}\;\ge\;C\;+\;\frac{1}{2\,\sigma_{\min}(\mathcal{R}_*)}\,
% \ln\!\left(
% \tfrac{1+\frac{1}{(1+\bar{E}_w)^2}}{\frac{1}{(1+\bar{E}_w)^2}-\xi}
% \right),
% \end{equation}
% \begin{equation}\label{eq:R-sufficient}
% \lambda_{\max}(R)\;\le\;\frac{\xi\,Q}{\;\| \mathcal{R}_*^\dagger A^\kappa \|^2}\,,
% \end{equation}
% \begin{equation}\label{eq:N-lb}
% N \ge \max\left\{
% 2,  
% 1+\tfrac{\gamma^2}{\Bigl(1+\tfrac{1}{(1+\bar{E}_w)^2}\Bigr)-\gamma}
% \right\}.
% \end{equation}
% Then, \eqref{condition_ISS} holds.

% \label{corollarysufficientcondition}
% \end{corollary}

\begin{remark}
    Corollary \ref{corollarysufficientcondition} says that if we penalize more on the states deviation compared with the inputs (i.e. small enough $\frac{\lambda_{\max}(R)}{Q}$) , and if the input constraint is large enough, then there exists long enough horizon that recovers the non-asymptotic bounds as in Theorem \ref{stabilitytheorem} after some time. Note that there are many ways in how to interpret Theorem \ref{stabilitytheorem}, and Corollary \ref{corollarysufficientcondition} is just one possibility.
\end{remark}

\section{Proof of Main Result}
In this section, we provide detailed proofs of Theorem \ref{stabilitytheorem} and all the auxiliary results needed. In Section \ref{estimationerrobound}, the non-asymptotic estimation error bound using RLS in (\ref{RLS}) is presented. Based on the error bound guarantee, we prove the stability of the closed-loop system in Section \ref{stabilitybound}. 

\subsection{Estimation Error Bound} \label{estimationerrobound}
In this section, the non-asymptotic estimation error bound for the RLS in (\ref{RLS}) is established. We first show that system (\ref{sys1}) with controller (\ref{controlall}) satisfies some excitation condition through the following lemma. 
 
\begin{lemma} 
Consider system~\eqref{sys1} under Assumption~\ref{assump1} and controller~\eqref{eq:controlpolicyperstep}. 
There exist constants $c_{\mathrm{PE}}>0$ and $p_{\mathrm{PE}}>0$ such that, 
for all $x \in \mathbb{R}^n$, $\zeta \in \mathcal{S}^{n+\kappa m-1}, s\in \{\mathrm{MPC},\mathrm{DB}\}$, and $\theta \in \mathbb{R}^{n\times (n+m)}$,
\[
\mathbb{P}\left(
  \Big| \zeta^\top[(x+w);(\pi_s(x+w,\theta)+v)] \Big|^2
  \geq c_{\mathrm{PE}}
\right) \geq p_{\mathrm{PE}},
\]
where $w \overset{d}{\sim} w_t$ and $v \overset{d}{\sim} v_t$.
\label{lemma_BMSB}
\end{lemma}

\begin{remark}
    The proof can be found in \cite{lei2026learningbasedhomothetictubempc}. The main proof idea is to consider the closed-loop system as a shifted version of sub-Gaussian random variables, represented by $\bar{v}_\tau,\bar{w}_\tau$. Then, by leveraging the tail bound and the lower bound of the variance in Assumption \ref{assump1}, an anti-concentration inequality can be derived for the closed-loop system. Our proof is partially inspired by \cite{li2023non}, but does not require boundedness on either the state variables or the process noise. Also, as opposed to existing results, e.g. \cite{siriya2023stability}, we can retrieve explicit values of $c_{\mathrm{PE}},p_{\mathrm{PE}}$.
\end{remark}

To establish a non-asymptotic error bound for parameter identification for each timestep $t$, a few definitions are required as introduced below. The high probability process noise bound and state norm bound are defined as \cite{siriya2024non}
\[
\overline{w}(t,\delta) :=
\begin{cases}
\sigma_w \sqrt{2n \ln\left( \frac{n \pi^2 t^2}{3\delta} \right)}, & t \in \mathbb{N}, \\
0, & t = 0,
\end{cases}
\]
\[
\overline{x}(t,\delta, x_0) := |x_0|+\|B\| \, t \, u_{\max} + \sum_{k=0}^{t-1}\,\overline{w}(t,\delta)
\]
The Gramian upper bound is defined as 
\begin{equation}
    \beta_{\max}(t,x_0)=\sum_{i=1}^t \left[ \bar{x}^2(t,\delta,x_0) + u_{\max}^2 \right] + \psi
    \label{eq:betamax}
\end{equation}
with $\psi$ from (\ref{RLS}). Denote $\Phi(T,\delta,x_0)
\triangleq \frac{2}{(1-\ln 2)\,p_{\mathrm{PE}}}\!\Bigg[
\ln\!\frac{\pi^2}{2\delta}
+ 2\ln T + (n+m) \,\ln\!\Bigg(1+\frac{16}{c_{\mathrm{PE}}p_{\mathrm{PE}}(T-1)}
\sum_{i=1}^{T}\big(\bar{x}^2(i,\delta,x_0)+u_{\max}^2\big)\Bigg)
\Bigg]+1$, and the burn-in time is $T_{\mathrm{burn\text{-}in}}(\delta,x_0)
\triangleq \min\Big\{T\in\mathbb{N}\, \Big|\, T\ge \Phi(T,\delta,x_0)\Big\}.$ We can conclude $T_{\text{burn-in}}(\delta,x_0) < \infty$ because $\Phi(T,\delta,x_0)=O(\ln(T))$ which grows slower than $O(T)$.
% Define the burn-in time as $T_{\mathrm{burn\text{-}in}}(\delta,x_0)\triangleq\inf\Big\{T\in\mathbb{N}\ \big|\ t\ge \tfrac{2}{(1-\ln2)p_{\mathrm{PE}}}\big[ \ln\tfrac{\pi^2(t-T+1)^2}{2\delta}+(n+m)\ln\!\big(1+\tfrac{16\sum_{i=1}^t \left[ \bar{x}^2(t,\delta, x_0) + u_{\max}^2 \right]}{c_{\mathrm{PE}}p_{\mathrm{PE}}(t-1)}\big)\big]+1,\ \forall t\ge T \Big\}$.

Using the above definitions, we have the following results.
\begin{proposition} (\cite{siriya2024non})
    For system (\ref{sys1}) with Assumptions \ref{assump1}-\ref{assump2}, under the RLS (\ref{RLS}) and controller (\ref{eq:controlpolicyperstep}), then for any $\delta \in (0,1)$ and initial state $x_0 \in \mathbb{R}^n$, $T_{\text{burn-in}}(\delta,x_0) < \infty$. Besides, we have 
    \[
        \mathbb{P}\left( \lVert \hat{\theta}_t - \theta^*  \rVert \leq \epsilon_t, \forall t \geq T_{\text{burn-in}}(\delta, x_0) \right) \geq 1-\delta
    \]
    where the error bound is defined as 
    \begin{equation}
    \begin{split}
        \epsilon_t &:= \epsilon(t,\delta,x_0) = \frac{\psi^{1/2} \lVert \theta^* \rVert_F}{\sqrt{\frac{c_{\mathrm{PE}}p_{\mathrm{PE}}}{4}(t -1)}+\psi} \\
        & + \frac{\sigma_W \sqrt{2n\ln(\tfrac{3n}{\delta})+\tfrac{n+m}{2}\ln(\tfrac{\beta_{\max}(t,x_0)}{\psi})} }{\sqrt{\frac{c_{\mathrm{PE}}p_{\mathrm{PE}}}{4}(t-1)}+\psi}
    \end{split}
    \end{equation}
    with $\psi$ from (\ref{RLS}), $\beta_{\max}$ from \eqref{eq:betamax} and $c_{\mathrm{PE}},p_{\mathrm{PE}}$ from Lemma \ref{lemma_BMSB}.

    \label{non-asymp bound}
\end{proposition}
\begin{proof}
    Proposition \ref{non-asymp bound} can be proved directly using Lemma \ref{lemma_BMSB}, Corollary 1 and Example 2 of \cite{siriya2024non}. 
\end{proof}    

% \begin{remark}
%     There are two ways to appreciate the non-asymptotic guarantees in Proposition \ref{non-asymp bound}. Firstly, for any given $\delta, x_0$, the error bound decreases monotonically with the sample size $t$ at a rate of $O(\frac{1}{\sqrt{t}})$. On the other hand, for any given error bound level $e$ and $x_0$, the probability that the error bound holds increases monotonically with the sample size $t$.
% \end{remark}

\subsection{Stability Bound} \label{stabilitybound}

In this section, using the non-asymptotic error bound in Section \ref{estimationerrobound}, we prove that the closed-loop system is practically stable as in Theorem \ref{stabilitytheorem}. There are two key steps: 1) we prove that as long as the prediction horizon is long enough, the disturbance-free version of (\ref{sys1}) can be stabilized by solving the MPC problem (\ref{MPCProblem}) without terminal constraints (Lemma \ref{MPCwithoutterminal}); 2) Secondly, we show the existence of a Convex-Decay Stochastic ISS Function (defined later) for the closed-loop system using a perturbation analysis on the optimal value function of the MPC problem (\ref{MPCProblem}) ; 3) Lastly, leveraging an existing result from \cite{siriya2025framework}, we show that the existence of a Convex-Decay Stochastic ISS Function leads to the bound in Theorem \ref{stabilitybound}.

We now give two definitions related to stability. The first one is the RPI set, describing the region where the system states remain. Building on this, a Convex-Decay Stochastic ISS Function is defined under the influence of noises and estimation error.
\begin{definition}
    The set \(\mathcal{X}_{\text{RPI}} \subseteq \mathbb{X}\) is said to be a \emph{robustly positive invariant (RPI)} set with error constant \(\bar{\theta} > 0\) for the system (\ref{sys1_subsampled}) under the influence of the policy \(\alpha\) if, for all \(x \in \mathcal{X}_{\text{RPI}}\), \(\bar{w} \in \mathbb{W}^\kappa\), \(\bar{v} \in \mathbb{V}^\kappa\), and $\hat{\theta}$ satisfying \(\|\hat{\theta}-\theta^*\|\leq \bar{\theta}\), it holds that:
\[
A^\kappa x + \mathcal{R}_* \alpha(x, \bar{v}, \hat{\theta})+ \mathcal{R}(A,I)\bar{w}\in \mathcal{X}_{\text{RPI}}.
\]
For convenience, we say that system (\ref{sys1}) with policy \(\alpha\) is \((\mathcal{X}_{\text{RPI}}, \bar{\theta}, w, v)\)-RPI.

\end{definition}
\begin{definition} (Convex-Decay Stochastic ISS Function)
    A Borel measurable function \(V : \mathcal{X}_{\text{RPI}} \to \mathbb{R}_{\geq 0}\), defined over the domain \(\mathcal{X}_{\text{RPI}} \subseteq \mathbb{X}\), is called a Convex-Decay Stochastic ISS Function with an estimation error \(\tilde{\theta}\) for the closed-loop dynamics $g:=g(x, w, \alpha(x, v, \tilde{\theta} + \theta_{*}))$ under the influence of the policy $\alpha(x, v, \tilde{\theta} + \theta_{*})$ and the stochastic noise \(w, v\), if:

1. System (\ref{sys1}) with policy \(\alpha\) is \((\mathcal{X}_{\text{RPI}}, \bar{\theta}, w, v)\)-RPI.

2. There exist class-\(\mathcal{K}_\infty\) functions \(\alpha_1, \alpha_2, \alpha_3\), class-$\mathcal{K}$ function \(\sigma_3\), a constant \(\tilde{d} \geq 0\) and $\bar{\theta}>0$, such that the following conditions hold for all \(x \in \mathcal{X}_{\text{RPI}}\) and $\|\tilde{\theta}\| \leq \bar{\theta}$:

\[
\alpha_1(|x|) \leq V(x) \leq \alpha_2(|x|),
\]
\[
\mathbb{E}[V(g) ] - V(x) \leq -\alpha_3(|x|) + \tilde{d} + \sigma_3(|\tilde{\theta}|),
\]
where \(\alpha_3 \circ \alpha_2^{-1}\) is lower-bounded by some convex function in \(\mathcal{K}_\infty\) for all $x \in \mathcal{X}_{\text{RPI}}$.
\label{SLF}
\end{definition}

Before proceeding to the analysis of the closed-loop system \ref{sys1}, two intermediate lemmas are derived. Lemma \ref{lem:lyasat} says that for known 1-step reachable linear system, the same finite horizon cost as in \eqref{MPCProblem} evaluated at the closed-loop trajectory under the saturated deadbeat controller serves as a Lyapunov function. Similarly, Lemma \ref{MPCwithoutterminal} says that for known 1-step reachable linear system under the MPC control policy, the finite-horizon optimal cost function serves as a Lyapunov function of the closed-loop system if the prediction horizon is long enough. This is enabled by leveraging the continuity, detectability, and controllability of the cost function. 
% Unlike widely used quadratic costs, a special cost in exponential form is implemented, which results in a uniform lower bound for the prediction horizon that guarantees stability of the closed-loop system. 

\begin{lemma}
    For system $\bar{x}_{\tau+1}=A^\kappa \bar{x}_\tau + \mathcal{R}_* \bar{u}_\tau$ under Assumption \ref{assump1} and the control policy $\pi_{\mathrm{DB}}(x,\theta^*,C)$ for any $u_{\max}-C>0$ and $x \in \mathbb{R}^n$, if $N > 1 + \frac{\big(Q+\lambda_{\max}(R)\,\|\mathcal{R}_*^\dagger A^\kappa\|^2\big)^2}{(1-p^2)^2\,Q^2} $, with $p$ being defined in \eqref{eq:p}, then the cost function 
    \begin{equation}
\begin{split}
    Y_{\mathrm{DB}}(x) &= \sum_{i=0}^{N-1}\ell(\hat{x}_i,\pi_i) + \ell_f(\hat{x}_N),\\
\text{s.t.} \quad & \hat{x}_0 = x, \\
& \hat{x}_{i+1}=A^\kappa \hat{x}_i + \mathcal{R}_* \pi_{\mathrm{DB}}(\hat{x}_i,\theta_*,C) \\
\end{split} \label{eq:optimalcostsat}
\end{equation}
    is a Lyapunov function of the closed-loop system.
    \label{lem:lyasat}
\end{lemma}
\begin{proof}
    Firstly, we establish the sandwich bounds for $Y_{\mathrm{DB}}(x)$. Define the following closed-loop evolution $\bar{x}_{\tau+1}=g^{\mathrm{sat}}(\bar{x}_\tau) :=  A^\kappa \bar{x}_\tau + \mathcal{R}_* \text{sat}_{u_{\max}-C}(-\mathcal{R}_*^\dagger A^\kappa \bar{x}_\tau)$. From the properties of the saturated deadbeat controller, we have
    \begin{equation}
    \label{temp2}
    \begin{aligned}
    &e^{|g^{\mathrm{sat}}(\bar{x}_\tau)|} - 1 \\
    &\leq \max \big\{0, \exp\big(\|A\|^\kappa |\bar{x}_\tau| - \sigma_{\min}(\mathcal{R}_*)(u_{\max}-C)\big) - 1 \big\} \\
    &\leq \max \big\{0, \exp\big(|\bar{x}_\tau| - \sigma_{\min}(\mathcal{R}_*)(u_{\max}-C)\big) - 1 \big\} \\
    & \leq \exp\big(|\bar{x}_\tau| - \sigma_{\min}(\mathcal{R}_*)(u_{\max}-C)\big) - 1\\
    &\leq \exp\big(-\sigma_{\min}(\mathcal{R}_*)(u_{\max}-C)\big) \,\big(e^{|\bar{x}_\tau|}-1\big) \\
    &=: p\,\big(e^{|\bar{x}_\tau|} - 1\big),
    \end{aligned}
    \end{equation}
    where $0 < p < 1$. The first inequality is from the consideration of two situations: when the deadbeat controller is not saturated, the system reaches $0$ in the next step; otherwise, we have the term $\|A\|^\kappa\,|\bar{x}_\tau| - \sigma_{\min}(\mathcal{R}_*)(u_{\max}-C)>0$. The second inequality is because $\|A\|\leq 1$ in Assumption \ref{assump2}. The last inequality follows from asic algebraic manipulations. We also have the following bound:
    \begin{equation}
    \begin{aligned}
        &|\pi^\intercal_{\mathrm{DB}}(\bar{x}_\tau,\theta^*,C) \,R\, \pi_{\mathrm{DB}}(\bar{x}_\tau,\theta^*,C)| \\
        & \quad \leq \lambda_{\max}(R) \, |\pi_{\mathrm{DB}}(\bar{x}_\tau,\theta^*,C)|^2 \\
        &  \quad \leq \lambda_{\max}(R) \, \max \{|\mathcal{R}_*^\dagger A^\kappa  \bar{x}_\tau|^2, (u_{\max}-C)^2 \} \\
        & \quad \leq \lambda_{\max}(R) \,|\mathcal{R}_*^\dagger A^\kappa  \bar{x}_\tau|^2 \\
        &  \quad \leq \lambda_{\max}(R) \|\mathcal{R}_*^\dagger A^\kappa\|^2  |\bar{x}_\tau|^2 \\
        &  \quad \leq \lambda_{\max}(R) \|\mathcal{R}_*^\dagger A^\kappa\|^2  \sigma(\bar{x}_\tau)
    \end{aligned}
    \label{eq:piDBbound}
    \end{equation}
    where the second inequality considers the cases when the deadbeat controller is saturated and unsaturated respectively, and the third inequality follows from the saturation definition of $\pi_\mathrm{DB}(\bar{x}_\tau,\theta^*,C)$. The last inequality is because $\sigma(x) \geq |x|^2$.
    
    For $\bar{x}_\tau := (g^{\mathrm{sat}})^{\tau}(\bar{x}_0)$ which is the closed-loop state obtained after $\tau$ successive applications of $g^{\mathrm{sat}}$, starting from $\bar{x}_0$, we have
    \begin{equation}
    \sigma(g^{\mathrm{sat}}(\bar{x}_\tau)) \le p^2 \sigma(\bar{x}_\tau), \ 
    \sigma(\bar{x}_\tau) \le p^{2\tau} \sigma(\bar{x}_0).
    \label{eq:sigma_decreasebound}
    \end{equation}
    Combining \eqref{eq:piDBbound}\eqref{eq:sigma_decreasebound} and the definitions of the stage and terminal costs, we have
    \begin{equation}
    \label{temp3}
    \begin{aligned}
    & \ell(\bar{x}_\tau, \pi_{\mathrm{DB}}(\bar{x}_\tau, \theta^*,C)) \\
    & \quad \overset{\eqref{eq:piDBbound}}{\le} Q \sigma(\bar{x}_\tau) + \lambda_{\max}(R)\|\mathcal{R}_*^\dagger A^\kappa\|^2 \sigma(\bar{x}_\tau) \\
    & \quad \overset{\eqref{eq:sigma_decreasebound}}{\le} \big(Q + \lambda_{\max}(R)\|\mathcal{R}_*^\dagger A^\kappa\|^2\big) p^{2\tau} \sigma(\bar{x}_0), \\
    &\ell_f(x_N) = \ell(x_N, \mathbf{0}_{\kappa m}) \\
    & \quad \le \big(Q + \lambda_{\max}(R)\|\mathcal{R}_*^\dagger A^\kappa\|^2\big) p^{2N} \sigma(\bar{x}_0).
    \end{aligned}
    \end{equation}
    Taking the sum from $\tau=0$ to infinity of the above geometric sequence, together with the positive definiteness of the stage costs, the following holds for any $N \in \mathbb{N}$:
    \begin{equation}
    \begin{aligned}
        Q\sigma(x) &\leq Y_{\mathrm{DB}}(x)\leq \tfrac{(Q + \lambda_{\max}(R)\| \mathcal{R}_*^\dagger A^\kappa \|^2) \, \sigma(x) }{1-p^2}. 
    \end{aligned}
    \label{eq:sandboundVsat}
    \end{equation}
 
    The decrease condition is established below, following the same flow as Variant 1 of \cite{grune2012nmpc}. For the closed-loop evolution of the system $\bar{x}_\tau$ defined above, abbreviate the stage costs as $\ell_\tau:=\ell(\bar{x}_\tau,\pi_{\mathrm{DB}}(\bar{x}_\tau,\theta^*,C))$. This implies that $Y_{\mathrm{DB}}(\bar{x}_0) \geq  \sum_{t=0}^{N-1}\ell_\tau$. Denoting $\gamma:= \tfrac{(Q + \lambda_{\max}(R)\| \mathcal{R}_*^\dagger A^\kappa \|^2)}{Q(1-p^2)}$ as in Theorem \ref{stabilitytheorem}, then the upper bound (\ref{eq:sandboundVsat}) implies that $Y_{\mathrm{DB}}(\bar{x}_0) \leq \gamma Q \sigma(\bar{x}_0)$ and thus $\ell_q \leq \frac{\gamma}{N} Q\sigma(\bar{x}_0)$ for at least one $q \in \{ 0,...,N-1\}$. Since $\gamma > 1$, the condition in the premise of this lemma $N > 1+\gamma^2 \;\Rightarrow\; N > \gamma$, so $\ell_q \leq \frac{\gamma}{N} Q\sigma(\bar{x}_0)<Q\sigma(\bar{x}_0)$ which does not hold for $q=0$. Therefore, $q \geq 1$. On the other hand, we have $Q\sigma(\bar{x}_q) \leq \ell_q$ due to positive-definiteness, so
    \begin{equation}
        \sigma(\bar{x}_q) \leq \frac{\gamma}{N} Q\sigma(\bar{x}_0).
        \label{eq:temp7}
    \end{equation}
    Combining with (\ref{eq:sandboundVsat}), we also have 
    \begin{equation}
    \begin{aligned}
        \sum_{t=0}^{N-q}\ell_{q+t} + \ell_{N+1} & \leq Y_{\mathrm{DB}}(\bar{x}_q) \leq \gamma Q \sigma(\bar{x}_q)\\
        Q \sigma(\bar{x}_q) &\leq \ell_q \leq \frac{\gamma}{N} Q\sigma(\bar{x}_0)
    \end{aligned}
    \label{eq:temp6}
    \end{equation}
    for the same $q$. Then
    \[
    \begin{aligned}
        Y_{\mathrm{DB}}(\bar{x}_1) & = \sum_{t=1}^{q-1}\ell_\tau + \sum_{t=0}^{N-q}\ell_{q+t} + \ell_f(\bar{x}_{N+1})\\
        &\leq \sum_{t=1}^{q-1}\ell_\tau + \sum_{t=0}^{N-q}\ell_{q+t} + \ell_{N+1} \\
        & = \sum_{t=1}^{q}\ell_\tau - \ell_q + \sum_{t=0}^{N-q}\ell_{q+t} + \ell_{N+1}\\
        & \overset{\eqref{eq:temp6}}{\leq} \sum_{t=1}^{q}\ell_\tau - Q \sigma(\bar{x}_q) + \gamma Q \sigma(\bar{x}_q) \\
        & \leq \sum_{t=1}^{q}\ell_\tau + Q (\gamma - 1)\sigma(\bar{x}_q) \\
        & \overset{\eqref{eq:temp7}}{\leq} Y_{\mathrm{DB}}(\bar{x}_0)- \ell_0 + \frac{\gamma(\gamma-1)}{N}Q \sigma(\bar{x}_0) \\
        & \leq Y_{\mathrm{DB}}(\bar{x}_0)- Q \sigma(\bar{x}_0) + \frac{\gamma(\gamma-1)}{N}Q \sigma(\bar{x}_0)
    \end{aligned}
    \]
    By definition of $\bar{x}_1 = g^{\mathrm{sat}}(\bar{x}_0)$, we finally have
    \begin{equation}
    \begin{split}
        &Y_{\mathrm{DB}}(g^{\mathrm{sat}}(\bar{x}_0))-Y_{\mathrm{DB}}(\bar{x}_0) = Y_{\mathrm{DB}}(\bar{x}_1)-Y_{\mathrm{DB}}(\bar{x}_0) \\
        &\leq - Q \sigma(\bar{x}_0) + \frac{\gamma(\gamma-1)}{N}Q \sigma(\bar{x}_0) \\
        & \leq - Q \sigma(\bar{x}_0) + \frac{\gamma^2}{N}Q \sigma(\bar{x}_0) \\
        & \leq -Q\,\sigma(\bar{x}_\tau)+ \frac{\big(Q+\lambda_{\max}(R)\|\mathcal{R}_*^\dagger A^\kappa\|^2\big)^2\,\sigma(\bar{x}_\tau)}{(1-p^2)^2\,Q\,(N-1)}
    \end{split} \label{eq:onestepdecrease_sat}
    \end{equation}
    Since $\bar{x}_\tau$ is arbitrary, then if 
    \(
N > \frac{\big(Q+\lambda_{\max}(R)\,\|\mathcal{R}_*^\dagger A^\kappa\|^2\big)^2}{Q^2\,(1-p^2)^2}+1,
\)
we have 
\(
Y_{\mathrm{DB}}(g^{\mathrm{sat}}(x))-Y_{\mathrm{DB}}(x) < 0, \forall x \in \mathbb{R}^n.
\)
\end{proof}

\begin{lemma}
    For system $\bar{x}_{\tau+1}=A^\kappa \bar{x}_\tau + \mathcal{R}_* \bar{u}_\tau$ under Assumption \ref{assump1} and the control policy $\pi_{\mathrm{MPC}}(x,\theta^*)$ derived from solving the following optimization problem
\begin{equation}
\label{eq:OptimalCost}
\begin{aligned}
Y_{\mathrm{MPC}}(x) \triangleq & \min_{\{\pi_i\}_{i=0}^{N-1}}
\left[ \sum_{i=0}^{N-1} \ell(\hat{x}_i, \pi_i) + \ell_f(\hat{x}_N) \right], \\
\text{s.t.} \quad 
& \hat{x}_0 = x, \\
& \hat{x}_{i+1}=A^\kappa \hat{x}_i + \mathcal{R}_* \pi_i, \\
& \|\pi_i\|_\infty \leq u_{\max} - C
\end{aligned}
\end{equation}
    for any $u_{\max}-C>0$ and $x \in \mathbb{R}^n$, if $N > 1 + \frac{\big(Q+\lambda_{\max}(R)\,\|\mathcal{R}_*^\dagger A^\kappa\|^2\big)^2}{(1-p^2)^2\,Q^2} $ with $p$ being defined in \eqref{eq:p}, then $Y_{\mathrm{MPC}}(x)$ is a Lyapunov function of the closed-loop system.
    \label{MPCwithoutterminal}
\end{lemma}
\begin{proof}
From the positive-definiteness of the stage costs, we have $Y_{\mathrm{MPC}}(x) \geq Q\sigma(x), \forall \, x \in \mathbb{R}^n$. By the optimality property of MPC, any feasible control sequence together with its realized cost provides an upper bound on the optimal value function. Consequently, if a feasible solution can be constructed, its corresponding realized cost must be greater than or equal to the optimal cost $Y_{\mathrm{MPC}}$. In particular, the saturated deadbeat controller yields a feasible solution to the MPC problem \eqref{eq:OptimalCost}. Moreover, its associated realized cost has been derived in \eqref{temp3}. Therefore, we obtain
\begin{equation}
\label{eq:sandboundVMPC}
\begin{split}
    Q\sigma(x) & \le Y_{\mathrm{MPC}}(x) \le Y_{\mathrm{DB}}(x) \le \tfrac{\big(Q+\lambda_{\max}(R)\|\mathcal{R}_*^\dagger A^\kappa\|^2\big)\,\sigma(x)}{1-p^2}.
\end{split}
\end{equation}

With the sandwich bound above, we now show a decrease condition. We define the closed-loop transition map $g^{\mathrm{MPC}}(\bar{x}_\tau) := A^\kappa \bar{x}_\tau + \mathcal{R}_* \pi_{\mathrm{MPC}}(\bar{x}_\tau, \theta^*)$, representing the one-step evolution of the system with initial condition $\bar{x}_\tau$ under MPC policy $\pi_{\mathrm{MPC}}(\bar{x}_\tau, \theta^*)$ with the true parameters $\theta^*$. The following conditions hold, allowing us to apply Theorem~1 of \cite{grimm2005model}:

\begin{enumerate}
    \item[(i)] The stage cost $\ell(x,u)$ is continuous in its arguments.
    
    \item[(ii)]  The function $\sigma(x)$ satisfies the detectability condition with 
    $W(x)=0$, $\bar{\alpha}_W(\cdot)=0 \cdot \mathrm{Id}$, 
    $\alpha_W = Q \cdot \mathrm{Id}$, and $\gamma_W=\mathrm{Id}$ such that
    \[
    \begin{aligned}
    W(\bar{x}_\tau) &\leq \bar{\alpha}_W(\sigma(\bar{x}_\tau)), \\
    W(g^{\mathrm{MPC}}) - W(\bar{x}_\tau) &\le -\alpha_W(\bar{x}_\tau) + \gamma_W(\ell(\bar{x}_\tau,\bar{u}_\tau)).
    \end{aligned}
    \]

\end{enumerate}
Therefore, with (i)-(ii) above and \eqref{eq:sandboundVMPC}, applying Theorem~1 of \cite{grimm2005model}, we obtain
\begin{equation}
\begin{aligned}
Y_{\mathrm{MPC}}&(g^{\mathrm{MPC}}(\bar{x}_\tau)) - Y_{\mathrm{MPC}}(\bar{x}_\tau) 
 \\
&\quad \le -Q\,\sigma(\bar{x}_\tau) + \tfrac{\big(Q+\lambda_{\max}(R)\|\mathcal{R}_*^\dagger A^\kappa\|^2\big)^2\,\sigma(\bar{x}_\tau)}
{(1-p^2)^2\,Q\,(N-1)}.
\end{aligned} \label{Y1decreasetrueparam}
\end{equation}
 Since $\bar{x}_\tau$ is arbitrary, given that $N > 1 + \tfrac{\big(Q+\lambda_{\max}(R)\,\|\mathcal{R}_*^\dagger A^\kappa\|^2\big)^2}{(1-p^2)^2\,Q^2}$ as in the premise, it follows that 
\(
Y_{\mathrm{MPC}}(g^{\mathrm{MPC}}(x)) - Y_{\mathrm{MPC}}(x) < 0, x \in \mathbb{R}^n.
\)
\end{proof}

We state the following lemma, which establish the semi-global Lipschitz continuity of MPC controller in the system parameters.

\begin{lemma} [Semi-Global Lipschitz Continuity of MPC]
    For each compact set $\mathbb{X}_l =\{ x \mid |x| \leq \bar{x}_s \}$ with a constant $\bar{x}_s>0$, there exists a positive constants $\bar{\epsilon}_1>0$ and $D_{\mathrm{MPC}}(\bar{x}_s) < \infty$ such that $\forall x \in \mathbb{X}_l, \forall  \hat{\theta}$ satisfying $\|\hat{\theta} - \theta^*\| \leq \bar{\epsilon}_1,$ 
    \begin{equation}
        | \pi_{\mathrm{MPC}}(x, \hat{\theta}) - \pi_{\mathrm{MPC}}(x, \theta^*) | \leq D_{\mathrm{MPC}}(\bar{x}_s) \, \|\hat{\theta} - \theta^* \|.  \label{lipschMPC}
    \end{equation}
    Besides, $D_{\mathrm{MPC}}(\bar{x}_s)$ is non-decreasing in $\bar{x}_s$.
    \label{localxlipschitzMPC}
\end{lemma}

\begin{lemma}
    Consider system \eqref{sys1} under Assumption \ref{assump1} and the certainty-equivalence MPC control policy $\pi_{\mathrm{MPC}}(x,\hat{\theta})$ by solving \eqref{MPCProblem}. For each constant $0<\bar{x}_s<+\infty$, and for all $\|\hat{\theta}-\theta^*\| \leq \bar{\epsilon}_1$ with $\bar{\epsilon}_1, \bar{x}_s$ defined in Lemma \ref{localxlipschitzMPC}, the optimal value function $Y_{\mathrm{MPC}}(x)$ in \eqref{eq:OptimalCost} satisfies the following for all $|x|\leq \bar{x}_s$ and any $\mu_1 > 0$:
    \begin{equation}
    Q\sigma(x) \le Y_{\mathrm{MPC}}(x) \le \tfrac{\big(Q+\lambda_{\max}(R)\|\mathcal{R}_*^\dagger A^\kappa\|^2\big)\,\sigma(x)}{1-p^2}
    \label{eq:sandboundVMPC_CE} 
    \end{equation}
    \begin{equation}
    \begin{aligned}
        &\mathbb{E}[Y_{\mathrm{MPC}}(g_1(x))] - Y_{\mathrm{MPC}}(x) \leq -Q\sigma(x) + \sigma_{\mathrm{MPC}}
\end{aligned} \label{eq:Y1decrease_CE}
    \end{equation}
    where 
    \begin{align*}
    & \sigma_{\mathrm{MPC}}:= \sigma_{\mathrm{MPC}}(x,\theta^*,\hat{\theta},Q,R,\bar{w},\bar{v},\mu_1,\bar{x}_s) \\
        & = \left\{ \left[ \tfrac{(Q + \lambda_{\max}(R)\| \mathcal{R}_*^\dagger A^\kappa \|^2)  }{1-p^2} (E_{w1}+1)^2\right]m_2^2 \right\} \\
        & \quad \times \sigma(x) +  \tfrac{(Q + \lambda_{\max}(R)\| \mathcal{R}_*^\dagger A^\kappa \|^2)^2 }{Q\max(N-1,1)(1-p^2)^2} \\
        & \quad \times \sigma(x) + \mu_1 \tfrac{(Q + \lambda_{\max}(R)\| \mathcal{R}_*^\dagger A^\kappa \|^2)   }{1-p^2} \\
        & \quad \times 8m^2_2\sigma(x) + \tfrac{(Q + \lambda_{\max}(R)\| \mathcal{R}_*^\dagger A^\kappa \|^2)   }{1-p^2} E_{w1}^2  \\
        & \quad + \tfrac{1}{\mu_1}\tfrac{(Q + \lambda_{\max}(R)\| \mathcal{R}_*^\dagger A^\kappa \|^2)   }{1-p^2}(E_{w1}^4+E_{w1}^2)
    \end{align*}
    with 
    \[
    \begin{aligned}
        g_1(x)&:= A^\kappa x+\mathcal{R}_*\pi_{\mathrm{MPC}}(x,\hat{\theta})+\mathcal{R}_*\bar{v}+ \mathcal{R}_\kappa(A,I)\bar{w}, \\
        h&:= \ln\left[ \mathbb{E} e^{|\mathcal{R}_* \bar{v}|}\right] +\ln\left[ \mathbb{E} e^{| \mathcal{R}_\kappa(A,I)\bar{w}|} \right],\\
        E_{w1} &:= \exp\Bigl(D_{\mathrm{MPC}}(\bar{x}_s)\|\mathcal{R}_*\| \|\hat{\theta} - \theta^*\| + |h|\Bigr) - 1 ,\\
        m^2_2 & = -1 +\tfrac{(Q + \lambda_{\max}(R)\| \mathcal{R}_*^\dagger, A^\kappa \|^2)^2 }{Q^2\max(N-1,1) (1-p^2)^2}  + \tfrac{(Q + \lambda_{\max}(R)\| \mathcal{R}_*^\dagger A^\kappa \|^2)  }{Q(1-p^2)}, \\
            & \bar{w} \overset{d}{\sim} \bar{w}, \bar{v} \overset{d}{\sim} \bar{v}.
    \end{aligned}
    \]
    
    \label{lem:MPCwithoutterminal_CE}
\end{lemma}
\begin{proof}
    The sandwich bound \eqref{eq:sandboundVMPC_CE} has been proved in \eqref{eq:sandboundVMPC}. We focus now on the proof of \eqref{eq:Y1decrease_CE}. Define and notice the difference between three closed-loop evolutions 
        \begin{equation}
        \begin{split}
            g_1(x)&:= A^\kappa x+\mathcal{R}_*\pi_{\mathrm{MPC}}(x,\hat{\theta})\\
            & \qquad\qquad+\mathcal{R}_*\bar{v}+ \mathcal{R}_\kappa(A,I)\bar{w}, \\
            \hat{g}_1(x)&:=  A^\kappa x+\mathcal{R}_*\pi_{\mathrm{MPC}}(x, \hat{\theta}), \\
            \tilde{g}_1(x)&:= A^\kappa x+\mathcal{R}_*\pi_{\mathrm{MPC}}(x, \theta^*),
        \end{split}
        \end{equation}
        which represent different next states under the true parameters, with or without stochastic disturbance, and using the ideal or CE controller. We focus on the analysis of the following value
        \begin{equation}
        \begin{split}
        &\mathbb{E}[Y_{\mathrm{MPC}}(g_1(x))] - Y_{\mathrm{MPC}}(x) \\
        &= \mathbb{E}[Y_{\mathrm{MPC}}(g_1(x))] - Y_{\mathrm{MPC}}(\tilde{g}_1(x)) \\
        & \qquad + Y_{\mathrm{MPC}}(\tilde{g}_1(x))  - Y_{\mathrm{MPC}}(x).  
        \end{split} \label{perturbationY1}
        \end{equation}
        
        Firstly, we bound $\mathbb{E}[Y_{\mathrm{MPC}}(g_1(x))] - Y_{\mathrm{MPC}}(\tilde{g}_1(x))$. We have
            \begin{equation}
            \label{eq:g1_g1tilde}
            \begin{split}
                &\mathbb{E} [\sigma(g_1(x))]\\
                &  = \mathbb{E} [\sigma(\hat{g}_1(x) + \mathcal{R}_*\bar{v}+ \mathcal{R}_\kappa(A,I)\bar{w})] \\
                & \leq \mathbb{E} [\sigma(\tilde{g}_1(x) + \mathcal{R}_*(\pi_s(x, \hat{\theta}) - \pi_s(x, \theta^*)) \\
                & \qquad\qquad\qquad+ \mathcal{R}_*\bar{v}+ \mathcal{R}_\kappa(A,I)\bar{w})] \\
                & = \mathbb{E} [\sigma(|\tilde{g}_1(x) + \mathcal{R}_*(\pi_s(x, \hat{\theta}) - \pi_s(x, \theta^*)) \\
                & \qquad\qquad\qquad+ \mathcal{R}_*\bar{v}+ \mathcal{R}_\kappa(A,I)\bar{w}|)] \\
                & \leq \mathbb{E} [\sigma(|\tilde{g}_1(x)| + D_{\mathrm{MPC}}(\bar{x}_s) \|\mathcal{R}_*\| \|\hat{\theta} -   \theta^* \|  + |h|)] \\
                & = \sigma(|\tilde{g}_1(x)| + D_{\mathrm{MPC}}(\bar{x}_s) \|\mathcal{R}_*\| \|\hat{\theta} -   \theta^* \|  + |h|)\\
                & =\sigma(\tilde{g}_1(x)) + E_{w1}^2 + \sigma(\tilde{g}_1(x))E_{w1}^2 \\
                & \quad + 2(e^{|\tilde{g}_1(x)|}-1)E_{w1} + 2\sigma(\tilde{g}_1(x))E_{w1} \\
                & \quad+ 2(e^{|\tilde{g}_1(x)|}-1)E_{w1}^2
            \end{split}
            \end{equation}
            where the second inequality is due to Jensen’s inequality and Lemma \ref{localxlipschitzMPC}. The Jensen’s inequality holds for $\sigma(\cdot)$ because it is convex. Next, we derive the relationship between $\sigma(\tilde{g}_1(x))$ and $\sigma(x)$. We have
                \begin{equation}
                \begin{split}
                    &Q\sigma(\tilde{g}_1(x)) - \tfrac{(Q + \lambda_{\max}(R)\| \mathcal{R}_*^\dagger A^\kappa \|^2)  \sigma(x) }{1-p^2} \\
                    & \overset{\eqref{eq:sandboundVMPC}}{\leq} Y_{\mathrm{MPC}}(\tilde{g}_1(x))  - Y_{\mathrm{MPC}}(x)  \\
                    & \overset{\eqref{Y1decreasetrueparam}}{\leq} - Q\sigma(x)+\tfrac{(Q + \lambda_{\max}(R)\| \mathcal{R}_*^\dagger A^\kappa \|^2)^2 \sigma(x)}{Q\max(N-1,1) (1-p^2)^2}.
                \end{split} 
                \end{equation}
            By rearranging, we arrive at
            \begin{equation}
            \begin{split}
                \sigma(\tilde{g}_1(x))
                &\leq - \sigma(x)+\tfrac{(Q + \lambda_{\max}(R)\| \mathcal{R}_*^\dagger A^\kappa \|^2)^2 \sigma(x)}{Q^2\max(N-1,1) (1-p^2)^2}  \\
                & \qquad \qquad + \tfrac{(Q + \lambda_{\max}(R)\| \mathcal{R}_*^\dagger A^\kappa \|^2)  \sigma(x) }{Q(1-p^2)} \\
                & = m^2_2 \, \sigma(x). 
            \end{split}
            \label{eq:relationm2}
            \end{equation}
            By inspection, $m_2 \geq 0$. Therefore, we have
            \begin{equation}
            \label{eq:temp2}
                \begin{split}
                    &\mathbb{E}[Y_{\mathrm{MPC}}(g_1(x))] - Y_{\mathrm{MPC}}(\tilde{g}_1(x)) \\
                    &\overset{\eqref{eq:sandboundVMPC}}{\leq} \dfrac{(Q + \lambda_{\max}(R)\| \mathcal{R}_*^\dagger A^\kappa \|^2)  }{1-p^2} \mathbb{E}[\sigma(g_1(x))] - Q\sigma(\tilde{g}_1(x)) \\
                    & \overset{\eqref{eq:g1_g1tilde}\eqref{eq:relationm2}}{\leq} \left[ \dfrac{(Q + \lambda_{\max}(R)\| \mathcal{R}_*^\dagger A^\kappa \|^2)  }{1-p^2} (E_{w1}+1)^2 - Q \right] \\
                    & \times m_2^2 \,\sigma(x) + \dfrac{(Q + \lambda_{\max}(R)\| \mathcal{R}_*^\dagger A^\kappa \|^2)   }{1-p^2} \\
                    & \times \left[ E_{w1}^2 + 2m_2(e^{|x|}-1)E_{w1}+ 2m_2(e^{|x|}-1)E_{w1}^2\right]
                \end{split} 
                \end{equation}

            Secondly, we consider the bound of $Y_{\mathrm{MPC}}(\tilde{g}_1(x))  - Y_{\mathrm{MPC}}(x)$, which has been established in \eqref{Y1decreasetrueparam}.
            
            Finally, combining \eqref{eq:temp2} with \eqref{Y1decreasetrueparam} and inserting into \eqref{perturbationY1}:
            \begin{equation}
            \begin{split}
                &\mathbb{E}[Y_{\mathrm{MPC}}(g_1(x))] - Y_{\mathrm{MPC}}(x) \\
        &= \mathbb{E}[Y_{\mathrm{MPC}}(g_1(x))] - Y_{\mathrm{MPC}}(\tilde{g}_1(x)) + Y_{\mathrm{MPC}}(\tilde{g}_1(x))  - Y_{\mathrm{MPC}}(x) \\
        & \overset{\eqref{eq:temp2}}{\leq} \left[ \dfrac{(Q + \lambda_{\max}(R)\| \mathcal{R}_*^\dagger A^\kappa \|^2)  }{1-p^2} (E_{w1}+1)^2 - Q \right] \\
                    & \times m_2^2 \,\sigma(x) + \dfrac{(Q + \lambda_{\max}(R)\| \mathcal{R}_*^\dagger A^\kappa \|^2)   }{1-p^2} \\
                    & \times \left[ E_{w1}^2 + 2m_2(e^{|x|}-1)E_{w1}+ 2m_2(e^{|x|}-1)E_{w1}^2\right] \\
                    & -Q\,\sigma(x) + \frac{\big(Q+\lambda_{\max}(R)\|\mathcal{R}_*^\dagger A^\kappa\|^2\big)^2\,\sigma(x)}
{(1-p^2)^2\,Q\,(N-1)}.
            \end{split}
            \label{eq:temp3}
            \end{equation}
            Further applying Young's Inequality with any constant $\mu_1 > 0$ to get 
            \[
            \begin{aligned}
                & 2m_2(e^{|x|}-1)E_{w1}+ 2m_2(e^{|x|}-1)E_{w1}^2 \\
                & \quad\leq 8\mu_1 m_2^2 \sigma(x) + \frac{1}{\mu_1} (E_{w1}^2+E_{w1}^4)
            \end{aligned}
            \]
            and inserting it into \eqref{eq:temp3} gives us \eqref{eq:Y1decrease_CE}.
\end{proof}

In Lemma \ref{lem:MPCwithoutterminal_CE}, we have established the decrease condition of $Y_{\mathrm{MPC}}$. We now proceed to establish the decrease condition of $Y_{\mathrm{DB}}$ using a similar proof idea.
\begin{lemma} \cite[Lemma 1]{siriya2023stability}
Suppose Assumption \ref{assump2} holds. For each $u_{\max}>C$, there exist $\bar{\epsilon}_2, D_{\mathrm{DB}} > 0$ such that
\[
\left| \pi_{\mathrm{DB}}(x,\hat{\theta},C) - \pi_{\mathrm{DB}}(x,\theta^*,C) \right|
   \leq D_{\mathrm{DB}} \,\|\hat{\theta}-\theta^*\|
\]
holds for all $x \in \mathbb{R}^n$ and $\|\hat{\theta}-\theta^*\| \leq \bar{\epsilon}_2 $.
\label{lem:loballipschitzDB}
\end{lemma}

\begin{lemma}
    Consider system \eqref{sys1} under Assumption \ref{assump1} and the certainty-equivalence DB control policy $\pi_{\mathrm{DB}}(x,\hat{\theta},C)$, where $\|\hat{\theta}-\theta^*\| \leq \bar{\epsilon}_2$ with $\bar{\epsilon}_2$ being defined in Lemma \ref{lem:loballipschitzDB}. Consider optimal cost $Y_{\mathrm{DB}}(x)$ defined as \eqref{eq:optimalcostsat}, for all $x \in \mathbb{R}^n$ and any $\mu_1 > 0$, we have
    \begin{equation}
    Q\sigma(x) \le Y_{\mathrm{DB}}(x) \le \frac{\big(Q+\lambda_{\max}(R)\|\mathcal{R}_*^\dagger A^\kappa\|^2\big)\,\sigma(x)}{1-p^2}
    \label{eq:sandboundVDB_CE} 
    \end{equation}
    \begin{equation}
    \begin{aligned}
        &\mathbb{E}Y_{\mathrm{DB}}(g_2(x)) - Y_{\mathrm{DB}}(x)\leq -Q\sigma(x)+\sigma_{\mathrm{DB}}
\end{aligned} \label{eq:Y2decrease_CE}
    \end{equation}
    where 
    \begin{align*}
    &\sigma_{\mathrm{DB}}:=\sigma_{\mathrm{DB}}(x,\theta^*,\hat{\theta},Q,R,\bar{w},\bar{v},\mu_1) \\
        & = \left\{ \left[ \tfrac{(Q + \lambda_{\max}(R)\| \mathcal{R}_*^\dagger A^\kappa \|^2)  }{1-p^2} (E_{w2}+1)^2 -Q \right]m_2^2 \right\} \\
        & \quad \sigma(x) +  \tfrac{(Q + \lambda_{\max}(R)\| \mathcal{R}_*^\dagger A^\kappa \|^2)^2 }{Q\max(N-1,1)(1-p^2)^2}  \times \sigma(x) + \mu_1 \tfrac{(Q + \lambda_{\max}(R)\| \mathcal{R}_*^\dagger A^\kappa \|^2)   }{1-p^2} \\
        & \quad \times 8m^2_2\sigma(x) + \tfrac{(Q + \lambda_{\max}(R)\| \mathcal{R}_*^\dagger A^\kappa \|^2)   }{1-p^2} E_{w2}^2  \\
        & \quad + \tfrac{1}{\mu_1}\dfrac{(Q + \lambda_{\max}(R)\| \mathcal{R}_*^\dagger A^\kappa \|^2)   }{1-p^2}(E_{w2}^4+E_{w2}^2)
    \end{align*}
    with 
    \[
    \begin{aligned}
        g_2(x)&:= A^\kappa x+\mathcal{R}_*\pi_{\mathrm{DB}}(x,\hat{\theta},C) +\mathcal{R}_*\bar{v}+ \mathcal{R}_\kappa(A,I)\bar{w} \\
        E_{w2} &:= \exp\Bigl(D_{\mathrm{DB}}\, \|\mathcal{R}_*\| \|\hat{\theta} - \theta^*\| + |h|\Bigr) - 1,\\
        & \bar{w} \overset{d}{\sim} \bar{w}_\tau, \bar{v} \overset{d}{\sim} \bar{v}_\tau.
    \end{aligned}
    \]
    where $h, m_2$ are defined the same as in Lemma \ref{lem:MPCwithoutterminal_CE} and $p$ is defined in \eqref{eq:p}.
    
    \label{DB_CE}
\end{lemma}
\begin{proof}
    The proof follows the same step as in Lemma \ref{lem:MPCwithoutterminal_CE}, by using results from Lemmas \ref{lem:loballipschitzDB} and \ref{lem:lyasat}.
\end{proof}

Based on the above two lemmas, we next prove the existence of Convex-Decay Stochastic ISS Function of the closed-loop system. We will need Lemma \ref{lem:levelsetuncertainty} below.
\begin{lemma}
\label{lem:levelsetuncertainty}
    For each $\frac{1}{2}c_3 < c_4 < c_{\mathrm{MPC}}, u_{\max},Q,R,N$, there exists $\bar{\epsilon}_3$ such that if $\|\hat{\theta}-\theta^*\| \leq  \bar{\epsilon}_3$, we have
    \begin{align*}
        Y_{\mathrm{MPC}}(\hat{\theta},\bar{x}_\tau) \geq c_{\mathrm{MPC}} &\;\Rightarrow\; Y_{\mathrm{MPC}}(\bar{x}_\tau) \geq c_{\mathrm{MPC}} - c_4,\\
        Y_{\mathrm{MPC}}(\hat{\theta},\bar{x}_\tau) \leq c_{\mathrm{MPC}}-c_3 &\;\Rightarrow\; Y_{\mathrm{MPC}}(\bar{x}_\tau) \leq c_{\mathrm{MPC}} - c_3+c_4.
    \end{align*}

\end{lemma}

Now we show the existence of Convex-Decay Stochastic ISS Function.
\begin{proposition} (Existence of Convex-Decay Stochastic ISS Function)
\label{SLya1}
    Consider system (\ref{sys1_subsampled}) under the certainty-equivalence control law $\alpha(x,v,\hat{\theta})$ as defined in (\ref{controlall}). Select $c_{\mathrm{MPC}},c_{\mathrm{DB}},c_3$ according to \eqref{eq:hyswitchlevelsetconstants}, then there exists a continuously differentiable function $\rho \in \mathcal{K}_\infty$ such that, defining $\tilde{Y}_1:=\rho \circ Y_{\mathrm{MPC}}, \tilde{Y}_2:= Y_{\mathrm{DB}}$, for each $\delta \in (0,1)$, the function
    \[
    V(x,s_\tau) := \tilde{Y}_{s_\tau}(x)
    \]
    is a Convex-Decay Stochastic ISS Function for the closed-loop system if the following conditions hold:
    \begin{flalign}
    & q_3^{\max} \left[ \tfrac{(Q + \lambda_{\max}(R)\| \mathcal{R}_*^\dagger A^\kappa \|^2)  }{1-p^2} (1 + E_w)^2 -Q \right]m_2^2 \nonumber \\
    & \quad - q_3^{\min}Q + q_3^{\max}\tfrac{(Q + \lambda_{\max}(R)\| \mathcal{R}_*^\dagger A^\kappa \|^2)^2 }{Q\max(N-1,1)(1-p^2)^2} <0 \label{ineq:1} \\
    &\|\hat{\theta}-\theta^*\| \leq  \min\{\bar{\epsilon}_1, \bar{\epsilon}_2, \bar{\epsilon}_3 \} \label{ineq:2} 
    \end{flalign}
    where
    \begin{equation}
    \begin{split}
        E_w &:= \exp\Bigl((D_{\mathrm{DB}}+D_{\mathrm{MPC}}(\bar{x}_s) \|\mathcal{R}_*\| \|\hat{\theta} - \theta^*\| \\
        & \quad + |\ln\left[ \mathbb{E} e^{|\mathcal{R}_* \bar{v}_\tau|}\right] +\ln\left[ \mathbb{E} e^{| \mathcal{R}_\kappa(A,I) \bar{w}_\tau|} \right]|\Bigr) - 1,\\
        \bar{x}_s &= \ln \left( \sqrt{\frac{c_{\mathrm{MPC}}}{Q}}+1  \right).
    \end{split}
    \end{equation}
    The Lipschitz constants $D_{\mathrm{DB}},D_{\mathrm{MPC}}(\bar{x}_s)$ are from Lemmas \ref{lem:loballipschitzDB} and \ref{localxlipschitzMPC} respectively. The estimation error thresholds $\bar{\epsilon}_1, \bar{\epsilon}_2, \bar{\epsilon}_3$ are requirements inherited from Lemmas \ref{localxlipschitzMPC}, \ref{lem:loballipschitzDB} and \ref{lem:levelsetuncertainty}. And $\gamma,q_2,q_3^{\max},q_3^{\min}$ are defined the same as in Theorem \ref{stabilitytheorem}.
\end{proposition}
\begin{proof}
Firstly, we notice that condition \eqref{ineq:2} ensures that the premises of Lemmas \ref{localxlipschitzMPC}, \ref{lem:loballipschitzDB} and \ref{lem:levelsetuncertainty} all hold.

We will be using the following lemma for the proof:
\begin{lemma} [Lemma 2 of \cite{nesic2004framework}]
Given strictly positive real numbers $c_3, c_4$ and $c_{\text{MPC}}$ satisfying 
$c_3 < c_{\mathrm{MPC}}, \frac{1}{2}c_3 < c_4 < c_{\mathrm{MPC}}$ and class-$\mathcal{K}_\infty$ functions 
$\underline{\kappa}, \, \bar{\kappa}$, there exists a continuously differentiable 
function $\rho \in \mathcal{K}_\infty$, with its derivative $\rho'(s):=\frac{d\rho(s)}{ds}$ nondecreasing, such that
\begin{align}
    \rho(s) &\leq \underline{\kappa}(s) && \forall s \in [0, c_{\mathrm{MPC}} - c_3+c_4], \label{eq:lemma2_1}\\
    \bar{\kappa}(s) &\leq \rho(s) && \forall s \in [c_{\mathrm{MPC}}-c_4, \infty). \label{eq:lemma2_2}
\end{align} \label{lem:rho_nesic}
\end{lemma}

Continuing the proof of the proposition, we define
\begin{equation}
\underline{\kappa}(s):= \frac{2}{3 \gamma} s, \bar{\kappa}(s):= 2 \,\gamma \, s
\label{eq:kappadefinition}
\end{equation}
so that due to the sandwich bounds of $Y_{\mathrm{MPC}}, Y_{\mathrm{DB}}$ in \eqref{eq:sandboundVMPC} and \eqref{eq:sandboundVsat}, we have
\begin{equation}
\begin{split}
    \frac{3}{2}\underline{\kappa}(Y_{\mathrm{MPC}}(x))\leq Y_{\mathrm{DB}}(x) \leq \frac{1}{2}\bar{\kappa}(Y_{\mathrm{MPC}}(x)). 
\end{split}
\label{eq:kappainequality}
\end{equation}
Following the construction procedure of Lemma 2 of \cite{nesic2004framework}, we can construct $\rho$ satisfying Lemma \ref{lem:rho_nesic} with the following properties:
\begin{equation}
\begin{aligned}
    \frac{d \rho}{ds}(s) &=
\begin{cases}
\frac{1}{2\gamma}, s \in [0, c_{\mathrm{MPC}} - c_3+c_4] \\
q_{2}, \,s \in [c_{\mathrm{MPC}}-c_4, \infty) \\
\tfrac{q_{2} - \frac{1}{2\gamma}}{c_3-2c_4} 
(s - c_{\mathrm{MPC}} + c_3-c_4) + \frac{1}{2\gamma}, 
\text{else}
\end{cases} \\
q_2 &= \max\left\{\frac{4\, \gamma \, c_{\mathrm{MPC}}}{c_3}, \frac{1}{2 \gamma}, 2\gamma   \right\} \geq 2
\end{aligned}
\label{eq:rhodefinition}
\end{equation}
where those three conditions do not overlap due to the premises of Lemma \ref{lem:rho_nesic}. It is easily seen that $\rho$ is a convex function with non-decreasing derivative, so for any $s_1, s_2 \geq 0$, we have
\begin{equation}
    \rho(s_1) \leq \rho(s_2) + \rho'(s_1) (s_1-s_2),
    \label{eq:rhoconvexinequality}
\end{equation}
and also 
\begin{equation}
    \frac{1}{2\gamma} \leq \rho'(s) \leq q_2.
    \label{eq:rhoderivativebounds}
\end{equation}
Bsed on the above properties, the sandwich bounds for $V(x,s)$ is therefore trivial, by combining \eqref{eq:rhoderivativebounds}, \eqref{eq:sandboundVMPC} and \eqref{eq:sandboundVsat}:
\begin{equation}
    \underbrace{\frac{1}{2 \gamma} Q\sigma(x)}_{\alpha_1(|x|)}\leq V(x,s) \leq \underbrace{q_2 \gamma Q \sigma(x)}_{\alpha_2(|x|)}.
\end{equation}
Denote the following $\kappa$-step evolutions of the system under MPC and deadbeat controller as
\[
\begin{aligned}
g_1(\bar{x}_\tau)&:= A^\kappa \bar{x}_\tau+\mathcal{R}_*\pi_{\mathrm{MPC}}(\bar{x}_\tau,\hat{\theta})+\mathcal{R}_* \bar{v}_\tau+ \mathcal{R}_\kappa(A,I) \bar{w}_\tau \\
g_2(\bar{x}_\tau)&:= A^\kappa \bar{x}_\tau+\mathcal{R}_*\pi_{\mathrm{DB}}(\bar{x}_\tau,\hat{\theta},C)+\mathcal{R}_* \bar{v}_\tau+ \mathcal{R}_\kappa(A,I) \bar{w}_\tau \\
\end{aligned}
\]
respectively, and denote the overall closed-loop system as 
\[
g(\bar{x}_\tau):= A^\kappa\bar{x}_\tau  + \mathcal{R}_* \alpha(\bar{x}_\tau,\bar{v}_\tau,\hat{\theta})+\mathcal{R}_\kappa(A,I) \bar{w}_\tau.
\]
We now consider the following four cases of switching:

(1) $Y_{\mathrm{MPC}}(\hat{\theta},\bar{x}_\tau) \geq c_{\mathrm{MPC}}, s_\tau \neq \text{DB}$. We have from Lemma \ref{lem:levelsetuncertainty} that this leads to $Y_{\mathrm{MPC}}(\bar{x}_\tau)\geq c_{\mathrm{MPC}}-c_4$. This case is when the system is about to escape to $s_{\tau+1}=\text{DB}$ with $s_\tau=\text{MPC}$. We have
\[
\begin{aligned}
    & \mathbb{E}[V(g(\bar{x}_\tau),s_{\tau+1})] - V(\bar{x}_\tau,s_\tau) \\
    & = \mathbb{E} [Y_{\mathrm{DB}}(g_2(\bar{x}_\tau))] - \rho(Y_{\mathrm{MPC}}(\bar{x}_\tau)) \\
    & \overset{\eqref{eq:Y2decrease_CE}}{\leq} Y_{\mathrm{DB}}(\bar{x}_\tau)-Q\sigma(\bar{x}_\tau)- \rho(Y_{\mathrm{MPC}}(\bar{x}_\tau))+ \sigma_{\mathrm{DB}}\\
    & \overset{\eqref{eq:kappainequality}}{\leq} \frac{1}{2}\bar{\kappa}(Y_{\mathrm{MPC}}(\bar{x}_\tau)) -Q\sigma(\bar{x}_\tau)- \rho(Y_{\mathrm{MPC}}(\bar{x}_\tau))+ \sigma_{\mathrm{DB}}\\
    & \overset{\eqref{eq:lemma2_2}}{\leq} -\frac{1}{2}\rho(Y_{\mathrm{MPC}}(\bar{x}_\tau))-Q\sigma(\bar{x}_\tau) + \sigma_{\mathrm{DB}}\\
    & \overset{\eqref{eq:sandboundVMPC}}{\leq} -\frac{1}{2}\rho \circ (Q\sigma(\bar{x}_\tau))-Q\sigma(\bar{x}_\tau) + \sigma_{\mathrm{DB}} \\
    & \overset{\eqref{eq:rhodefinition}}{\leq} -\left(1+\frac{1}{4 \gamma} \right)Q\sigma(\bar{x}_\tau) + \sigma_{\mathrm{DB}}\\
    & \leq \left(1+\frac{1}{4 \gamma} \right)\big[-Q\sigma(\bar{x}_\tau)+ \sigma_{\mathrm{DB}}\big].
\end{aligned}
\]
where the last inequality uses the fact that $1+\frac{1}{4 \gamma}>1$ and $\sigma_{\mathrm{DB}}>0$.

(2) $Y_{\mathrm{DB}}(\hat{\theta},x_\tau) \leq c_{\mathrm{DB}}, s_\tau \neq \text{MPC}$. Due to the optimality of MPC, we have $Y_{\mathrm{MPC}}(\hat{\theta},x_\tau) \leq Y_{\mathrm{DB}}(\hat{\theta},x_\tau)$. In addition to condition \eqref{eq:hyswitchlevelsetconstants}, we have
    \[
        Y_{\mathrm{DB}}(\hat{\theta},x_\tau) \leq c_{\mathrm{DB}} \;\Rightarrow\; Y_{\mathrm{MPC}}(\hat{\theta},x_\tau) \leq c_{\mathrm{MPC}}-c_3,
    \]
which, according to Lemma \ref{lem:levelsetuncertainty}, implies that 
\[
Y_{\mathrm{MPC}}(\bar{x}_\tau) \leq c_{\mathrm{MPC}}-c_3 + c_4
\]
and so $\rho(s) = \frac{1}{2\gamma}s$ according to \eqref{eq:rhodefinition}, which is linear. Using this fact, we have
\begin{equation}
\begin{split}
    & \mathbb{E}[V(g(\bar{x}_\tau),s_{\tau+1})] - V(\bar{x}_\tau,s_\tau) \\
    & = \mathbb{E} [\rho \circ Y_{\mathrm{MPC}}(g_1(\bar{x}_\tau))] - Y_{\mathrm{DB}}(\bar{x}_\tau) \\
    & \overset{\eqref{eq:Y1decrease_CE}}{\leq} \rho\circ (Y_{\mathrm{MPC}}(\bar{x}_\tau)-Q\sigma(\bar{x}_\tau) + \sigma_{\mathrm{MPC}})- Y_{\mathrm{DB}}(\bar{x}_\tau)  \\
    & \overset{\eqref{eq:rhoconvexinequality}}{\leq} \rho\circ (Y_{\mathrm{MPC}}(\bar{x}_\tau)) - Y_{\mathrm{DB}}(\bar{x}_\tau)\\
    & \quad + \rho'(Y_{\mathrm{MPC}}(\bar{x}_\tau)-Q\sigma(\bar{x}_\tau)+ \sigma_{\mathrm{MPC}}) \\
    & \quad \times \big[-Q\sigma(\bar{x}_\tau)+ \sigma_{\mathrm{MPC}}\big].\\
\end{split}
\label{eq:temp1}
\end{equation}
Firstly, we have
\[
\begin{aligned}
    & \rho\circ (Y_{\mathrm{MPC}}(\bar{x}_\tau)) - Y_{\mathrm{DB}}(\bar{x}_\tau) \\
    & \overset{\eqref{eq:kappainequality}}{\leq} \rho\circ (Y_{\mathrm{MPC}}(\bar{x}_\tau))  -\frac{3}{2}\underline{\kappa}(Y_{\mathrm{MPC}}(x)) \\
    & \overset{\eqref{eq:lemma2_1}}{\leq} -\frac{1}{2}\underline{\kappa}(Y_{\mathrm{MPC}}(x))   \overset{\eqref{eq:kappadefinition}\eqref{eq:sandboundVMPC}}{\leq} -\frac{1}{3 \gamma}Q \sigma(\bar{x}_\tau).
\end{aligned}
\]
Combining with \eqref{eq:rhoderivativebounds} and \eqref{eq:temp1}, if \(-Q\sigma(\bar{x}_\tau)+ \sigma_{\mathrm{MPC}} \geq 0,\) then 
\[
\begin{aligned}
    & \mathbb{E}[V(g(\bar{x}_\tau),s_{\tau+1})] - V(\bar{x}_\tau,s_\tau) \\
    & \overset{\eqref{eq:rhodefinition}}{\leq} -\frac{1}{3 \gamma}Q \sigma(\bar{x}_\tau) + q_2 \big[ -Q\sigma(\bar{x}_\tau) + \sigma_{\mathrm{MPC}} \big] \\
    & \leq q_2 \big[ -Q\sigma(\bar{x}_\tau) + \sigma_{\mathrm{MPC}} \big]
\end{aligned}
\]
Similarly, if \(-Q\sigma(\bar{x}_\tau)+ \sigma_{\mathrm{MPC}} < 0,\) then
\[
\begin{aligned}
    & \mathbb{E}[V(g(\bar{x}_\tau),s_{\tau+1})] - V(\bar{x}_\tau,s_\tau) \\
    & \overset{\eqref{eq:rhodefinition}}{\leq} -\frac{1}{3 \gamma}Q \sigma(\bar{x}_\tau) + \frac{1}{2 \gamma} \big[ -Q\sigma(\bar{x}_\tau)  + \sigma_{\mathrm{MPC}} \big] \\
    & \leq \frac{1}{2 \gamma}  \big[ -Q\sigma(\bar{x}_\tau) + \sigma_{\mathrm{MPC}} \big]
\end{aligned}
\]

(3) $s_\tau =s_{\tau+1}=\text{MPC}$. This case corresponds to when no switching occurs and the system evolves using the MPC controller. Thus we have
\[
\begin{aligned}
    & \mathbb{E}[V(g(\bar{x}_\tau),s_{\tau+1})] - V(\bar{x}_\tau,s_\tau) \\
    & = \mathbb{E} [\rho \circ Y_{\mathrm{MPC}}(g_1(\bar{x}_\tau))] - \rho\circ(Y_{\mathrm{MPC}}(\bar{x}_\tau)) \\
    & \overset{\eqref{eq:rhoconvexinequality}}{\leq} \mathbb{E} [\rho'(Y_{\mathrm{MPC}}(g_1(\bar{x}_\tau)))\cdot (Y_{\mathrm{MPC}}(g_1(\bar{x}_\tau)) - Y_{\mathrm{MPC}}(\bar{x}_\tau))] \\
    & \overset{\eqref{eq:Y1decrease_CE}}{\leq} \mathbb{E} [\rho'(Y_{\mathrm{MPC}}(g_1(\bar{x}_\tau)))\cdot (- Q\sigma(\bar{x}_\tau) + \sigma_{\mathrm{MPC}})].
\end{aligned}
\]
Using the mean value theorem, the monotonicity of $\rho'$ and the sandwich bounds of $\rho'$ in \eqref{eq:rhoderivativebounds}, we then have that when $- Q\sigma(\bar{x}_\tau) + \sigma_{\mathrm{MPC}}>0$,
\begin{align*}
    &\mathbb{E}[V(g(\bar{x}_\tau),s_{\tau+1})] - V(\bar{x}_\tau,s_\tau) \overset{\eqref{eq:rhoderivativebounds}}{\leq} q_2 (- Q\sigma(\bar{x}_\tau) + \sigma_{\mathrm{MPC}}). 
\end{align*}
And when $- Q\sigma(\bar{x}_\tau) + \sigma_{\mathrm{MPC}} \leq 0$,
\begin{align*}
    \mathbb{E}[V(g(\bar{x}_\tau),s_{\tau+1})] - V(\bar{x}_\tau,s_\tau)&\overset{\eqref{eq:rhoderivativebounds}}{\leq} \frac{1}{2\gamma}(- Q\sigma(\bar{x}_\tau) + \sigma_{\mathrm{MPC}}).
\end{align*}

(4) $s_\tau =s_{\tau+1}=\text{DB}$. This case corresponds to when no switching occurs and the system evolves using the deadbeat controller. Thus we have
\begin{equation}
\begin{split}
    &\mathbb{E}[V(g(\bar{x}_\tau),s_{\tau+1})] - V(\bar{x}_\tau,s_\tau)\\
    & = \mathbb{E} [Y_{\mathrm{DB}}(g_2(\bar{x}_\tau))] - Y_{\mathrm{DB}}(\bar{x}_\tau) \overset{\eqref{eq:Y2decrease_CE}}{\leq} -Q\sigma(\bar{x}_\tau) + \sigma_{\mathrm{DB}}.
\end{split}
\end{equation}

Before considering the above four cases, we need to derive an upper bound of $D_{\mathrm{MPC}}(\cdot)$. The MPC controller is activated when $ Y_{\mathrm{MPC}}(\hat{\theta},\bar{x}_\tau) \leq c_{\mathrm{MPC}} $. Together with the bound $Q\sigma(\bar{x}_\tau) \leq Y_{\mathrm{MPC}}(\hat{\theta},\bar{x}_\tau)$,we can set $\bar{x}_s := \ln \left( \sqrt{\frac{c_{\mathrm{MPC}}}{Q}}+1  \right)$ for the Lipschitz constant of MPC. 

Considering the above four cases together, remember that $q_3^{\min}=\min \left\{1+\frac{1}{4\gamma}, q_2, \frac{1}{2\gamma}, 1 \right\}, q_3^{\max}=\max \left\{1+\frac{1}{4\gamma}, q_2, \frac{1}{2\gamma}, 1 \right\}$. Then we have
\begin{equation}
\begin{split}
    &\mathbb{E}[V(g(\bar{x}_\tau),s_{\tau+1})] - V(\bar{x}_\tau,s_\tau) \\
    & \overset{\eqref{eq:Y1decrease_CE}\eqref{eq:Y2decrease_CE}}{\leq} - q_3^{\min} Q\sigma(\bar{x}_\tau) + q_3^{\max}\max \{
        \sigma_{\mathrm{DB}} , \sigma_{\mathrm{MPC}} \}\\
    & \leq \bigg\{ q_3^{\max}\left[  \tfrac{(Q + \lambda_{\max}(R)\| \mathcal{R}_*^\dagger A^\kappa \|^2)  }{1-p^2} (E_w+1)^2 - Q \right]m_2^2 \\  
        & \quad -q_3^{\min} Q\bigg\} \times \sigma(\bar{x}_\tau) +  q_3^{\max} \tfrac{(Q + \lambda_{\max}(R)\| \mathcal{R}_*^\dagger A^\kappa \|^2)^2 }{Q\max(N-1,1)(1-p^2)^2} \\
        & \quad \times \sigma(\bar{x}_\tau) + q_3^{\max} \mu_1 \tfrac{(Q + \lambda_{\max}(R)\| \mathcal{R}_*^\dagger A^\kappa \|^2)   }{1-p^2} \\
        & \quad \times 8m^2_2\sigma(\bar{x}_\tau) + q_3^{\max} \tfrac{(Q + \lambda_{\max}(R)\| \mathcal{R}_*^\dagger A^\kappa \|^2)   }{1-p^2} E_w^2  \\
        & \quad + \tfrac{q_3^{\max}}{\mu_1}\tfrac{(Q + \lambda_{\max}(R)\| \mathcal{R}_*^\dagger A^\kappa \|^2)   }{1-p^2}(E^4_w+E^2_w)
\end{split}
\end{equation}
The above inequality satisfies the second condition of Definition \ref{SLF} with 
        \[
        \begin{aligned}
            &-\alpha_3(|\bar{x}_\tau|) := \Big\{ q_3^{\max}\Big[ \tfrac{(Q + \lambda_{\max}(R)\| \mathcal{R}_*^\dagger A^\kappa \|^2)  }{1-p^2} (E_w+1)^2 \\
            & \quad -Q \Big]m_2^2 - q_3^{\min} Q \Big\} \sigma(\bar{x}_\tau)  +  \tfrac{q_3^{\max} (Q + \lambda_{\max}(R)\| \mathcal{R}_*^\dagger A^\kappa \|^2)^2 }{Q\max(N-1,1)(1-p^2)^2} \sigma(\bar{x}_\tau) \\
            & \quad + \mu_1 q_3^{\max}\tfrac{(Q + \lambda_{\max}(R)\| \mathcal{R}_*^\dagger A^\kappa \|^2)   }{1-p^2} 8m^2_2\sigma(\bar{x}_\tau) \\
            &\tilde{d} + \sigma_3(|\tilde{\theta}|) :=  q_3^{\max}\tfrac{(Q + \lambda_{\max}(R)\| \mathcal{R}_*^\dagger A^\kappa \|^2)   }{1-p^2} E_w^2  \\
            & \qquad \qquad \qquad \qquad+ \tfrac{q_3^{\max}}{\mu_1}\tfrac{(Q + \lambda_{\max}(R)\| \mathcal{R}_*^\dagger A^\kappa \|^2)   }{1-p^2}(E^4_w+E^2_w)
        \end{aligned}
        \]
        Since $\mu_1$ is arbitrary, if (\ref{ineq:1}) holds, then $\alpha_3(|x|)$ is a class $\mathcal{K}_\infty$ function with small enough $\mu_1$ satisfying
        \[
        \begin{aligned}
    & \mu_1 < \Bigg| q_3^{\max}\left[ \tfrac{(Q + \lambda_{\max}(R)\| \mathcal{R}_*^\dagger A^\kappa \|^2)  }{1-p^2} (1 + E_w)^2 -Q \right]m_2^2 - q_3^{\min} Q \nonumber \\
    & \quad + \tfrac{q_3^{\max}(Q + \lambda_{\max}(R)\| \mathcal{R}_*^\dagger A^\kappa \|^2)^2 }{Q\max(N-1,1)(1-p^2)^2}\Bigg| \, \bigg / \Bigg[q_3^{\max}\tfrac{(Q + \lambda_{\max}(R)\| \mathcal{R}_*^\dagger A^\kappa \|^2)   }{1-p^2} 8m^2_2 \Bigg].
    \end{aligned}
        \]

        Lastly, we prove that \(\alpha_3 \circ \alpha_2^{-1}\) is lower-bounded by some convex function in \(\mathcal{K}_\infty\). 
        \[
\begin{aligned}
&(\alpha_3\circ \alpha_2^{-1})(r) \\
&=
\frac{r}{q_2\gamma Q}
\Bigg[
q_3^{\min}Q
-
q_3^{\max}\Bigg(
 \tfrac{Q+\lambda_{\max}(R)\|\mathcal{R}_*^\dagger A^\kappa\|^2}{1-p^2}
 (E_w+1)^2
 -Q
 \Bigg)m_2^2
\\
&\quad
-
\tfrac{
q_3^{\max}
\left(
Q+\lambda_{\max}(R)\|\mathcal{R}_*^\dagger A^\kappa\|^2
\right)^2
}{
Q\max(N-1,1)(1-p^2)^2
} 
-
\mu_1 q_3^{\max}
\tfrac{
Q+\lambda_{\max}(R)\|\mathcal{R}_*^\dagger A^\kappa\|^2
}{1-p^2}
8m_2^2
\Bigg].
\end{aligned}
\]
        Clearly, $(\alpha_3\circ\alpha_2^{-1})(r)$ is convex in $r$. If (\ref{ineq:1}) is satisfied, then there exists some $\mu_1>0$ small enough such that $\alpha_3 \circ \alpha_2^{-1}(r)$ is a $\mathcal{K}_\infty$ function. Since $\mu_1$ can be arbitrary, the conclusion follows. In addition, the above inequality holds for all $\bar{x}_\tau \in \mathbb{R}^n$, so that the closed-loop system has $\mathbb{R}^n = \mathcal{X}_{\mathrm{RPI}}$.
\end{proof}

Based on the above derivation, we propose the following $\kappa$-step state bound of the closed-loop system.
\begin{proposition}
Consider system~\eqref{sys1} under Assumptions~\ref{assump1}--\ref{assump2}, and the control laws in Algorithm~\ref{algoall} with switching thresholds $c_{\mathrm{MPC}},\,c_{\mathrm{DB}},\,c_3$ satisfying~\eqref{eq:hyswitchlevelsetconstants}. For any horizon $N \ge 2$, confidence level $\delta \in (0,1)$, initial state $x_0 \in \mathbb{R}^n$, and threshold $\frac{1}{2}c_3 < c_4 < c_{\mathrm{MPC}}$, if \eqref{condition_ISS} holds, then there exists a constant $c \ge 0$, and a function $\eta \in \mathcal{L}$ such that
\begin{equation}
\label{eq:kappa_step_state_bound}
    \mathbb{P}\bigl(\,|\bar{x}_\tau| \le \eta(\tau) + c\,\bigr) \ge 1 - \delta, \forall \tau \ge \left\lceil \frac{t_1}{\kappa} \right\rceil,
\end{equation}
where 
\(
t_1 = t_1\!\bigl(\delta, x_0, A, B, Q, R, N, c_{\mathrm{MPC}}, c_{\mathrm{DB}}, c_3, c_4, u_{\max}, C\bigr).
\)

\label{stabilityproposition}
\end{proposition}
\begin{proof}
    Firstly, by Proposition \ref{non-asymp bound}, for each user assigned $\delta \in (0,1)$ and initial state $x_0 \in \mathbb{R}^n$, there must exist a time step $t_3(\delta,x_0,A,B,Q,R,N,u_{\max},c_{\mathrm{MPC}},c_3,c_4)$ such that \eqref{ineq:2} holds for all $t \geq t_3(\delta,x_0,A,B,Q,R,N,u_{\max},c_{\mathrm{MPC}},c_3,c_4)$ with probability at least $1-\delta$. The dependence on $c_{\mathrm{MPC}},c_3,c_4$ is inherited from Lemma \ref{lem:levelsetuncertainty}. Next, we show that if (\ref{condition_ISS}) is satisfied, then there exists a time $t_4(\delta,x_0,A,B,Q,c_{\mathrm{MPC}})$ such that (\ref{ineq:1}) starts to hold. Denote $\Delta_E = E_w - \bar{E}_w$ and $D=\max\{ D_{\mathrm{MPC}}(\bar{x}_s), D_{\mathrm{DB}}\} $. By inspection, there must exist a $\mathcal{K}_\infty$ function $\alpha_4(\cdot)$ such that 
    \begin{equation}
        |\Delta_E| \leq \alpha_4(D\,\|\hat{\theta}_{\kappa \tau}-\theta^*\|).
        \label{eq:temp5}
    \end{equation}
    Thus, (\ref{ineq:1}) is equivalent to 
\[
\begin{aligned}
    &q_3^{\max} \left[ \tfrac{(Q + \lambda_{\max}(R)\| \mathcal{R}_*^\dagger A^\kappa \|^2)  }{1-p^2} (1 + \bar{E}_w + \Delta_E)^2 -Q \right]m_2^2 \\
    & \quad - q_3^{\min}Q + q_3^{\max}\tfrac{(Q + \lambda_{\max}(R)\| \mathcal{R}_*^\dagger A^\kappa \|^2)^2 }{Q\max(N-1,1)(1-p^2)^2} < 0
\end{aligned}
\]
\begin{equation}
\label{eq:temp4}
\begin{aligned}
\Leftrightarrow &\tfrac{q_3^{\max} \left(Q + \lambda_{\max}(R)\,\| \mathcal{R}_*^\dagger A^\kappa \|^2\right)m_2^2}{1-p^2} \,\Delta_E^2 \\[4pt]
&\ + \tfrac{2q_3^{\max} \left(Q + \lambda_{\max}(R)\,\| \mathcal{R}_*^\dagger A^\kappa \|^2\right)\left(1+\bar{E}_w\right)m_2^2}{1-p^2} \,\Delta_E \\[4pt]
&\ + q_3^{\max}\left[ \tfrac{\left(Q + \lambda_{\max}(R)\,\| \mathcal{R}_*^\dagger A^\kappa \|^2\right)\left(1+\bar{E}_w\right)^2}{1-p^2} - Q \right] m_2^2 \\[4pt]
&\ - q_3^{\min}Q + \tfrac{q_3^{\max}\left(Q + \lambda_{\max}(R)\,\| \mathcal{R}_*^\dagger A^\kappa \|^2\right)^2}{Q\max(N-1,1)\,(1-p^2)^2} \;<\; 0
\end{aligned}
\end{equation}
Due to (\ref{condition_ISS}), there must exist some constant $d>0$ such that 
\[
\begin{aligned}
    &q_3^{\max}\left[ \tfrac{\left(Q + \lambda_{\max}(R)\,\| \mathcal{R}_*^\dagger A^\kappa \|^2\right)\left(1+\bar{E}_w\right)^2}{1-p^2} - Q \right] m_2^2 \\
&\ - q_3^{\min}Q + \tfrac{q_3^{\max}\left(Q + \lambda_{\max}(R)\,\| \mathcal{R}_*^\dagger A^\kappa \|^2\right)^2}{Q\max(N-1,1)\,(1-p^2)^2} \leq -d
\end{aligned}
\]
and that for \eqref{eq:temp4} to hold, it suffices to show that
\[
\begin{aligned}
&\tfrac{q_3^{\max} \left(Q + \lambda_{\max}(R)\,\| \mathcal{R}_*^\dagger A^\kappa \|^2\right)m_2^2}{1-p^2} \,\Delta_E^2 \\[4pt]
&+ \tfrac{2 q_3^{\max}\left(Q + \lambda_{\max}(R)\,\| \mathcal{R}_*^\dagger A^\kappa \|^2\right)\left(1+\bar{E}_w\right)m_2^2}{1-p^2} \,\Delta_E  \;<\; d.
\end{aligned}
\]
Combining with \eqref{eq:temp5}, it suffices to satisfy
\begin{equation}
\begin{aligned}
&\tfrac{q_3^{\max}\left(Q + \lambda_{\max}(R)\,\| \mathcal{R}_*^\dagger A^\kappa \|^2\right)m_2^2}{1-p^2} \ \alpha^2_4 \!\left(D\,\|\hat{\theta}_{\kappa \tau}-\theta^*\|\right) \\[4pt]
&\quad + \tfrac{2q_3^{\max}\left(Q + \lambda_{\max}(R)\,\| \mathcal{R}_*^\dagger A^\kappa \|^2\right)\left(1+\bar{E}_w\right)m_2^2}{1-p^2} \\
& \quad \times \alpha_4 \!\left(D\,\|\hat{\theta}_{\kappa \tau}-\theta^*\|\right)  \;<\; d
\end{aligned}
\label{eq:t_4}
\end{equation}
The LHS of the above condition is quadratic in $\alpha_4(\cdot)$ with positive coefficients. Combining with Proposition \ref{non-asymp bound}, there must exist a time $t_4(\delta,x_0,A,B,Q,c_{\mathrm{MPC}})$ such that for all $\kappa \tau \geq t_4$, $\alpha_4(\cdot)$ is small enough to satisfy the above inequality. The dependence of $t_4$ on $c_{\mathrm{MPC}}$ is inherited from that of $D$. 

Define $t_1(\delta,x_0,A,B,Q,R,N,u_{\max},c_{\mathrm{MPC}},c_3,c_4):= \max\{t_3,t_4 \}$. 
% \[
% \begin{aligned}
%     &t_1(\delta,x_0,A,B,Q,R,N,u_{\max},c_{\mathrm{MPC}},c_3,c_4) \\
%     &=\max
%     \begin{cases}
%         t_3(\delta,x_0,A,B,Q,R,N,u_{\max},c_{\mathrm{MPC}},c_3,c_4)\\
%         t_4(\delta,x_0,A,B,Q,c_{\mathrm{MPC}})\}
%     \end{cases}.
% \end{aligned}
% \]
Then from \cite[Theorem 2]{siriya2025framework}, by showing the existence of a Convex-Decay Stochastic ISS Function as in Proposition \ref{SLya1}, a high probability $\kappa$-step state bound \eqref{eq:kappa_step_state_bound} can be derived. The explicit form of $\eta, c$ are available from \cite[Remark 3]{siriya2025framework},which will then be used to construct $\tilde{\eta},\tilde{c}$ in Theorem \ref{stabilitytheorem}. Specifically, $\eta(\tau)$ is related to the estimation error, which decreases over time with high probability. And $c$ represents the non-vanishing influence of the process and excitory noises.
\end{proof}

\textit{Proof of Theorem \ref{stabilitytheorem}}. 

To continue with a per-step state bound with high probability, we first consider the deadbeat controller and the MPC controller separately before taking a union bound. We required the following lemma.
\begin{lemma} 
\label{lemma:kappato1step_dragan}
    For system \eqref{sys1}, if there exist constants $0<\delta<\frac{1}{3},c_1,c_2 \geq 0$, $r^s,r^b \in R_{>0}\cup\{\infty\}$, $\tau_0 \in \mathbb{N}$, $\beta, \beta_1 \in \mathcal{KL}$ and $\tilde{\gamma},\tilde{\gamma}_1,\gamma_1 \in \mathcal{K}_\infty$ such that
    \begin{equation}
    \begin{split}
    \mathbb{P}\Big(|\bar{x}_{\tau_0}| \leq r^s \;\Rightarrow\; |\bar{x}_{\tau}| &\leq \beta(|\bar{x}_{\tau_0}|,\tau-\tau_0)+c_1\Big), \\
    &  \geq 1-\delta, \forall \tau \geq \tau_0 \geq 0,
    \end{split}
    \label{eq:kappa_dragan1}
    \end{equation}
    \begin{equation}
    % \begin{split}
    %     \mathbb{P}\Big(\forall \tau_0 \geq 0, |\bar{x}_{\tau_0}| \leq r^b &\Rightarrow |x_{\tau \kappa + i}| \leq \tilde{\gamma}(|\bar{x}_{\tau_0}|)+c_2, \\
    %     &\forall i =0,...,\kappa-1\Big)\geq 1-2\delta,
    % \end{split}
    \begin{split}
\mathbb{P}\Big(
\forall &\tau \geq \tau_0,\,
|\bar{x}_{\tau}| \leq r^b
\Rightarrow
|x_{\tau \kappa + i}|
\leq
\tilde{\gamma}(|\bar{x}_{\tau}|)
\\
&\qquad +
\beta_1\!\left(
\epsilon_{\kappa \tau_0},
\kappa(\tau-\tau_0)
\right)
+c_2, \\
& \qquad \qquad \qquad \forall i=0,\ldots,\kappa-1
\Big)
\geq 1-2\delta .
\end{split}
    \label{eq:kappa_dragan2}
    \end{equation}
    then there exist $r^x \geq 0,c_3 \geq 0$ and $\bar{\beta} \in \mathcal{KL}$ such that
    \begin{equation}
    \begin{split}
        \mathbb{P}\Big(|\bar{x}_{\tau_0}| \leq r^x \;\Rightarrow\; |x_t| &\leq \bar{\beta}(|\bar{x}_{\tau_0}|+\epsilon_{\kappa\tau_0}, t - \tau_0 \kappa)+c_3\Big) \\
        & \geq 1-3\delta, \forall t \geq \tau_0 \kappa > 0.
    \end{split}
    \label{eq:kappa_dragan3}
    \end{equation}
\end{lemma}

\begin{proof}
We first prove that the determinisitc results hold. The proof is inspired by \cite[Theorem 1]{nevsic1999formulas}. Let \(t\geq \kappa\tau_0>0\). Then there exist unique
\(
\tau\geq \tau_0,\  i\in{0,\ldots,\kappa-1}
\)
such that
\(
t=\kappa\tau+i.
\)
Choose \(r^x>0\) such that
\(
r^x\leq r^s,
\ 
\beta(r^x,0)+c_1\leq r^b.
\)
This choice is immediate if \(r^b=\infty\); if \(r^b<\infty\), it requires the usual compatibility condition \(r^b>c_1\).

Assume that
\(
|\bar x_{\tau_0}|\leq r^x.
\)
Since \(r^x\leq r^s\), \eqref{eq:kappa_dragan1} gives
\(
|\bar x_{\tau}|
\leq
\beta(|\bar x_{\tau_0}|,\tau-\tau_0)+c_1.
\)
Using monotonicity of \(\beta(\cdot,\tau-\tau_0)\) and the fact that
\(
\beta(r^x,\tau-\tau_0)\leq \beta(r^x,0),
\)
we obtain
\(
|\bar x_{\tau}|
\leq
\beta(r^x,0)+c_1
\leq r^b.
\)
Therefore the premise of \eqref{eq:kappa_dragan2} is satisfied at this \(\tau\). Hence, 
\[
\begin{aligned}
|x_t|
&=
|x_{\kappa\tau+i}|
\leq
\tilde\gamma(|\bar x_\tau|)
+
\beta_1\!\left(
\epsilon_{\kappa\tau_0},
\kappa(\tau-\tau_0)
\right)
+c_2 \\                                            \
&\leq
\tilde\gamma\!\left(
\beta(|\bar x_{\tau_0}|,\tau-\tau_0)+c_1
\right)
+
\beta_1\!\left(
\epsilon_{\kappa\tau_0},
\kappa(\tau-\tau_0)
\right)
+c_2 .
\end{aligned}
\]
Using the elementary inequality
\(
\tilde\gamma(a+b)\leq \tilde\gamma(2a)+\tilde\gamma(2b),
 \forall a,b\geq0,
\)
we get
\[
\begin{aligned}
|x_t|
&\leq
\tilde\gamma\!\left(
2\beta(|\bar x_{\tau_0}|,\tau-\tau_0)
\right)
+
\beta_1\!\left(
\epsilon_{\kappa\tau_0},
\kappa(\tau-\tau_0)
\right)
+
\tilde\gamma(2c_1)+c_2 .
\end{aligned}
\]
Define
\(
c_3:=\tilde\gamma(2c_1)+c_2.
\)
Monotonicity gives
\[
\begin{aligned}
|x_t|
&\leq
\tilde\gamma\!\left(
2\beta(
|\bar x_{\tau_0}|+\epsilon_{\kappa\tau_0},
\tau-\tau_0
)
\right) 
+
\beta_1\!\left(
|\bar x_{\tau_0}|+\epsilon_{\kappa\tau_0},
\kappa(\tau-\tau_0)
\right)
+c_3 .
\end{aligned}
\]
Moreover,
\(
t-\kappa\tau_0
=
\kappa(\tau-\tau_0)+i,
\ i\in{0,\ldots,\kappa-1}.
\)
Thus the quantities \(\tau-\tau_0\) and \(t-\kappa\tau_0\) differ only by the fixed block length \(\kappa\) and integer $i$. From \cite[Lemma 3]{teel2000smooth}, there exists some
\(
\bar\beta\in\mathcal{KL}
\)
such that
\[
\begin{aligned}
&\tilde\gamma\!\left(
2\beta(
|\bar x_{\tau_0}|+\epsilon_{\kappa\tau_0},
\tau-\tau_0
)
\right)
+
\beta_1\!\left(
|\bar x_{\tau_0}|+\epsilon_{\kappa\tau_0},
\kappa(\tau-\tau_0)
\right)  \\                                                \
&\qquad\leq
\bar\beta\!\left(
|\bar x_{\tau_0}|+\epsilon_{\kappa\tau_0},
t-\kappa\tau_0
\right).
\end{aligned}
\]
Consequently,
\(
|x_t|
\leq
\bar\beta\!\left(
|\bar x_{\tau_0}|+\epsilon_{\kappa\tau_0},
t-\kappa\tau_0
\right)
+c_3.
\)

Now we continue to the proof of the probabilistic case. Define the events
\[
E_{1,\tau}
:=
\left\{
|\bar x_{\tau_0}|\le r^s
\;\Rightarrow\;
|\bar x_\tau|
\le
\beta(|\bar x_{\tau_0}|,\tau-\tau_0)+c_1
\right\},
\]
\[
\begin{aligned}
E_2
&:=
\bigcap_{\tau \ge \tau_0} \bigcap_{j=0}^{\kappa-1}
 \{
|\bar x_\tau|\le r^b
\;\Rightarrow\;
|x_{\tau\kappa+j}|
\le
\tilde\gamma\bigl(|\bar x_\tau|\bigr) \\
& \qquad +
\beta_1\!\left(
\epsilon_{\kappa \tau_0},
\kappa(\tau-\tau_0)\right)+ c_2
\}, \\
E_{3,t}
&:=
\left\{
|\bar x_{\tau_0}|\le r^x
\Rightarrow
|x_t|
\le
\bar\beta(|\bar{x}_{\tau_0}|+\epsilon_{\kappa\tau_0},t-\tau_0\kappa)+c_3
\right\}.
\end{aligned}
\]
By assumption,
\(
\mathbb P(E_{1,\tau})\ge 1-\delta,
\ 
\mathbb P(E_2)\ge 1-2\delta.
\)
The above proof has shown that
\(
E_{1,\tau}\cap E_2 \subseteq E_{3,t}.
\)
Therefore,
\[
\mathbb P(E_{3,t})
\ge
\mathbb P(E_{1,\tau}\cap E_2)
\ge
1-\mathbb P(E_{1,\tau}^c)-\mathbb P(E_2^c)
\ge
1-3\delta.
\]
Since $t\ge \tau_0\kappa$ was arbitrary, \eqref{eq:kappa_dragan3} holds for all fixed
$t\ge \tau_0\kappa$.
\end{proof}

Since we already have \eqref{eq:kappa_dragan1} satisfied through Proposition \ref{stabilityproposition}, it now suffices to prove that \eqref{eq:kappa_dragan2} holds. We will consider the case of deadbeat controller and MPC controller one by one. 

For the deadbeat controller, we can prove \eqref{eq:kappa_dragan2}. Due to the definition of $\pi_{\mathrm{DB}}(\cdot)$, we know that there exists state-dependent scaling coefficients $\alpha(x):\mathbb{R}^n \rightarrow (0,1]$ such that $\pi_{\mathrm{DB}}(x,\theta^*,C)=-\alpha(x)\mathcal R_\kappa^\dagger A^\kappa x$. For $0 \le i<\kappa$, unrolling the closed-loop dynamics under policy $\pi_{\mathrm{DB}}(x_{\tau \kappa},\hat{\theta}_{\tau \kappa},C)+\bar{v}_\tau$ gives for all $i=0,1,...,\kappa-1$:
\[
\begin{aligned}
    &|x_{\tau \kappa+i}|\\
    &=|A^i x_{\tau \kappa} + \sum_{s=0}^{i-1}A^{i-1-s}Bu_{\tau\kappa+s} + \sum_{s=0}^{i-1}A^{i-1-s}w_{\tau\kappa+s}| \\
&\leq \left| \left(A^i - [\alpha(\bar{x}_{\tau})\mathcal R_\kappa^\dagger A^\kappa]_{0:i-1}\right)x_{\tau \kappa}\right| \\
& \quad + \left|\sum_{s=0}^{i-1}A^{i-1-s}B \, D_{\text{DB}}\|\hat{\theta}_{\tau \kappa}-\theta^*  \|   \right| \\
& \quad + \left|\sum_{s=0}^{i-1}A^{i-1-s}B \, v_{\tau\kappa+s} \right| + \left|\sum_{s=0}^{i-1}A^{i-1-s}w_{t+s} \right| \\
& \leq |(1-\alpha(\bar{x}_{\tau}))A^i x_{\tau \kappa} + \alpha(\bar{x}_{\tau})\Big(A^i - [\mathcal R_\kappa^\dagger A^\kappa]_{0:i-1}\Big)x_{\tau \kappa}| \\
& \quad + \left|\sum_{s=0}^{i-1}A^{i-1-s}B D_{\text{DB}}\|\hat{\theta}_{\tau \kappa}-\theta^*  \|   \right|   + \left|\sum_{s=0}^{i-1}A^{i-1-s}B \, \sqrt{m} \, C \right| \\
& \quad + \sum_{s=0}^{i-1} \left| A^{i-1-s} \right| \sigma_W \sqrt{2n \ln\left( \frac{n \pi^2 i^2}{3 \delta} \right)} \text{ w.p. $1-\delta$}\\
\end{aligned}
\]
\[
\begin{aligned}
& \leq \underbrace{\max\{\|A^i\|, \big\|A^i - [\mathcal R_\kappa^\dagger A^\kappa]_{0:i-1} \big\| \} |x_{\tau \kappa}|}_{\gamma^{\mathrm{DB}}_{1,i}(|x_{\tau \kappa}|)} \\
& \quad + \underbrace{\left|\sum_{s=0}^{i-1}A^{i-1-s}B \, \sqrt{m} \, C \right|}_{\gamma_{v_1,i}(C)}   + \underbrace{\left|\sum_{s=0}^{i-1}A^{i-1-s}B \, D_{\text{DB}}\right| \, \epsilon_{\tau_0 \kappa}}_{\gamma^{\mathrm{DB}}_{\tilde{\theta}_1,i}(\epsilon_{\tau_0 \kappa})} \\
& \quad + \underbrace{\left|\sum_{s=0}^{i-1}A^{i-1-s} \right| \sigma_W \sqrt{2n \ln\left( \frac{n \pi^2 i^2}{3 \delta} \right)}}_{\gamma_{w_1,i}(\sigma_W)} \text{ w.p. $1-2\delta$},
\end{aligned}
\]
where the first inequality is from Lemma \ref{lem:loballipschitzDB}, and the second inequality with probability is from Assumption \ref{assump1} together with \cite[Lemma 9]{siriya2024non}. The last inequality is from Proposition \ref{non-asymp bound} and the fact taht $\tau \geq \tau_0$. So \eqref{eq:kappa_dragan2} is satisfied.

% Note that $\gamma_1, \gamma_{\tilde{\theta}_1}, \gamma_{v_1}, \gamma_{w_1}$ are class-$\mathcal{K}_\infty$ functions, and that $\gamma_{\tilde{\theta}_1}(\|\hat{\theta}_{\tau \kappa}-\theta^*  \|)$ can be upper bounded by a class-$\mathcal{L}$ function due to Proposition \ref{non-asymp bound}. 

% https://pweb.fbe.hku.hk/~pingyu/6066/LN/LN4_Maximum%20Theorem%2C%20Implicit%20Function%20Theorem%20and%20Envelope%20Theorem.pdf? Theorem 1
For MPC controller, according to Berge's Maximum Theorem \cite[Theorem 17.31]{aliprantis2006infinite} and the uniqueness of our MPC solution, the optimal control sequence to problem \(\mathbf{P}_N(x,\hat{A}, \hat{B})\) is continuous in $x$. We denote the $i^{th}$ CE MPC solution at time $t=\tau \kappa$ as $\hat{u}_{i|\tau \kappa}(x_{\tau \kappa}), i=0,1,...,\kappa-1$, and the MPC solution with the true parameters as $u_{i|\tau \kappa}^*(x_{\tau \kappa}), i=0,1,...,\kappa-1$. Since the MPC cost satisfies
\(
    Y_{\mathrm{MPC}}(x)
    \geq
    \sum_{j=0}^{N-1}
    \bigl(u^*_{j|\tau\kappa}(x)\bigr)^\top R u^*_{j|\tau\kappa}(x), R\succ 0,
\)
and since $Y_{\mathrm{MPC}}$ admits a class-$\mathcal K_\infty$ upper bound as in \eqref{eq:sandboundVMPC}, there exists a class $\mathcal{K}_\infty$ function $\alpha_5$ such that $|u_{i|\tau \kappa}^*(x)| \leq \alpha_5(|x|)$ for all $x \in \mathbb{X}_{\text{MPC}}$. Then the closed-loop dynamics when MPC controller is active becomes 
\[
\begin{aligned}
    &|x_{\tau \kappa+i}|=A^i x_{\tau \kappa} + \sum_{s=0}^{i-1}A^{i-1-s}B \, (\hat{u}_{s|\tau \kappa}(x_{\tau \kappa}) + v_{\tau\kappa+s})  \\
    & \quad+ \sum_{s=0}^{i-1}A^{i-1-s} w_{\tau\kappa+s} \\
    &\leq|A^i| | x_{\tau \kappa}| + \left|\sum_{s=0}^{i-1}A^{i-1-s}B \right| \left| \hat{u}_{s|\tau \kappa}(x_{\tau \kappa}) \right|  \\
    & \quad + \left| \sum_{s=0}^{i-1}A^{i-1-s}  w_{\tau\kappa+s} \right| + \left|\sum_{s=0}^{i-1}A^{i-1-s}B \, v_{\tau\kappa+s} \right| \\
&\leq \underbrace{|A^i| | x_{\tau \kappa}| + \left|\sum_{s=0}^{i-1}A^{i-1-s}B \right| \alpha_5(|x_{\tau \kappa}|)}_{\gamma^{\mathrm{MPC}}_{1,i}(|x_{\tau \kappa}|)} \\
& \quad + \underbrace{\left|\sum_{s=0}^{i-1}A^{i-1-s}B \, \sqrt{m} \, C \right|}_{\gamma_{v_1,i}(C)}  + \underbrace{\left|\sum_{s=0}^{i-1}A^{i-1-s}B \, D_{\text{MPC}}\right| \, \epsilon_{\tau \kappa}}_{\gamma^{\mathrm{MPC}}_{\tilde{\theta}_{1},i}(\epsilon_{\tau \kappa})} \\
& \quad + \underbrace{\sum_{s=0}^{i-1} \left|A^{i-1-s} \right| \sigma_W \sqrt{2n \ln\left( \frac{n \pi^2 i^2}{3 \delta} \right)}}_{\gamma_{w_1,i}(\sigma_W)}  \text{ w.p. $1-2\delta$},
\end{aligned}
\]

Finally, define the unified comparison functions by taking the maximum over all prediction indices $i=0,\ldots,\kappa-1$ and over both controller modes: \( \gamma_1(r) := \max_{\substack{i=0,\ldots,\kappa-1\\ q\in\{\mathrm{MPC},\mathrm{DB}\}}} \gamma^q_{1,i}(r), \) and \( \beta_1(r) := \max_{\substack{i=0,\ldots,\kappa-1\\ q\in\{\mathrm{MPC},\mathrm{DB}\}}} \gamma^q_{\tilde{\theta}_1,i}(r). \) Similarly, define \( \gamma_{v_1}(C) := \max_{i=0,\ldots,\kappa-1} \gamma_{v_1,i}(C), \gamma_{w_1}(\sigma_W) := \max_{i=0,\ldots,\kappa-1} \gamma_{w_1,i}(\sigma_W). \) Since the maxima are taken over finitely many class-$\mathcal K_\infty$ functions, the resulting functions remain of class $\mathcal K_\infty$. In addition, the MPC and deadbeat situations derived above all depend on the same events with probability $1-2\delta$. Hence, \eqref{eq:kappa_dragan2} holds by setting \( \tilde{\gamma}(s):=\gamma_1(s), c_2 := \gamma_{v_1}(C) + \gamma_{w_1}(\sigma_W), \) and \( r_s=r_b=r_x=+\infty. \). Finally, setting $\tau_0 = \left\lceil \frac{t_1}{\kappa} \right\rceil$ gives us the final conclusion.

% Finally, we take $\gamma_1(|x_{\tau \kappa}|):=\max\{\gamma^{\mathrm{MPC}}_1(|x_{\tau \kappa}|),\gamma^{\mathrm{DB}}_1(|x_{\tau \kappa}|) \}, \gamma_{\tilde{\theta}_1}(\|\hat{\theta}_{\tau \kappa}-\theta^*  \|) := \max\{ \gamma^{\mathrm{MPC}}_{\tilde{\theta}_1}(\|\hat{\theta}_{\tau \kappa}-\theta^*  \|), \gamma^{\mathrm{DB}}_{\tilde{\theta}_1}(\|\hat{\theta}_{\tau \kappa}-\theta^*  \|)\} $, and set $\tilde{\gamma}(s):=\gamma_1(s), c_2: = \gamma_{v_1}(C) + \gamma_{\tilde{\theta}_1}(\|\hat{\theta}_{\tau \kappa}-\theta^*  \|) + \gamma_{w_1}(\sigma_W), r_s=r_b=r_x = +\infty$. More detailed construction of the functions involved is available in \cite[Theorem 1]{nevsic1999formulas}.

\section{Simulation}
To demonstrate the effectiveness of the control method in Algorithm \ref{algoall}, we consider the following linear system
\[
x_{t+1} =
\begin{bmatrix}
\cos\frac{\pi}{4} & \sin\frac{\pi}{4} \\
-\sin\frac{\pi}{4} & \cos\frac{\pi}{4}
\end{bmatrix} x_t
+
\begin{bmatrix}
1 & 0.2 \\ -0.3 & 1
\end{bmatrix} u_t
+ w_t
\]
where $u_{\max}=6, w_t \overset{\mathrm{iid}}{\sim} \mathcal{N}(0,1), C=3, \hat{A}_0=\mathbf{I}_{2 \times 2}, \hat{B}_0=\mathbf{0}_{2 \times 2}, x_0=[20 \ -20]^\intercal, Q=1, R=0.1\times\mathbf{I}_{2 \times 2}, \lambda=0.1, N=5, c_{MPC} = 200, c_{DB} = 100, c_3 = 50$. We run the simulation for 8000 steps per episode, and 20 episodes are run with the same initial conditions to plot the ensemble averages of the states and the estimation error in Fig. \ref{estimationerror_ensemble}. With the true parameters, the DB control law and MPC problem are
\[
    \pi_{\mathrm{DB}}(x,\theta^*,C) =\text{sat}_{3}\left(-\begin{bmatrix}
0.80 & 0.53 \\
-0.47 & 0.87
\end{bmatrix} x \right),
\]
\[
\begin{aligned}
\min_{\{\pi_i\}_{i=0}^{4}} \quad & \sum_{i=0}^{4} \ell(\hat{x}_i, \pi_i) + \ell_f(\hat{x}_N), \\
\text{s.t.} \quad & \hat{x}_{0} = x, \\
& \hat{x}_{i+1}=\begin{bmatrix}
\cos\frac{\pi}{4} & \sin\frac{\pi}{4} \\
-\sin\frac{\pi}{4} & \cos\frac{\pi}{4}
\end{bmatrix} \hat{x}_i
+
\begin{bmatrix}
1 & 0.2 \\ -0.3 & 1
\end{bmatrix} \pi_i, \\
& \|\pi_i\|_\infty \leq 3.
\end{aligned}
\]

\begin{figure}[htbp]
\centering
\includegraphics[scale=0.1]{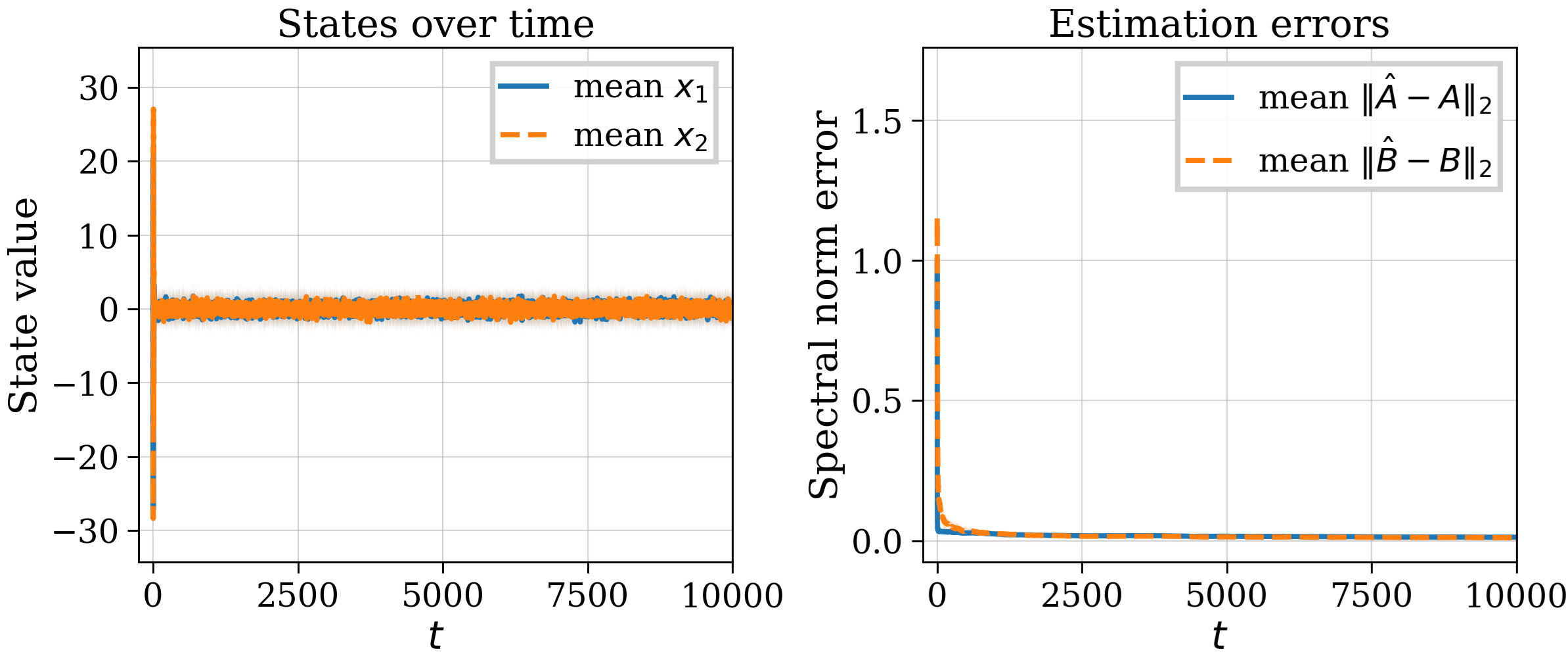}
\caption{States and estimation errors ensemble} \label{estimationerror_ensemble}
\end{figure}

\section{Conclusion}
In this paper, we propose an excited certainty equivalence Model Predictive Control scheme for adaptive control of discrete-time linear systems subject to additive i.i.d bounded stochastic disturbances with input constraints, assuming that the system is at least marginally stable. We proved the probabilistic PE condition for any process disturbance that obeys the sub-Gaussian description. In addition, a uniform lower bound on the prediction horizon is derived that ensures stability of the MPC without terminal constraints. A high probability practical non-asymptotic stability bound on the closed-loop system is established.

\bibliographystyle{IEEEtran}  
\bibliography{references}

\appendix

% \begin{figure}[htbp]
% \centering
% \includegraphics[scale=0.5]{theorem_dependence.PNG}
% \caption{Theorem dependence graph} \label{theorem_dependence}
% \end{figure}

\section*{Proof of Corollary \ref{corollarysufficientcondition}}
From~\eqref{eq:R-sufficient-new}, one has
\(
\frac{
Q+\lambda_{\max}(R)\|\mathcal{R}_*^\dagger A^\kappa\|^2
}{Q}
\le
1+\frac{\xi}{2}.
\)
Moreover, from~\eqref{eq:umax-sufficient-new},
\(
p^2
=
\exp\left(
-2\sigma_{\min}(\mathcal{R}_*)(u_{\max}-C)
\right)
\le
\frac{\xi}{2(1+\xi)}.
\)
Therefore,
\(
1-p^2
\ge
1-\frac{\xi}{2(1+\xi)}
=
\frac{2+\xi}{2(1+\xi)}.
\)
Hence
\(
\gamma
\le
\left(1+\frac{\xi}{2}\right)
\frac{2(1+\xi)}{2+\xi}
=
1+\xi .
\)
Since $\gamma\ge 1$, it follows that
\(
0\le \gamma-1\le \xi .
\)
Next, from~\eqref{eq:m2} and~\eqref{eq:gamma},
\(
m_2^2
=
\frac{\gamma^2}{N-1}+\gamma-1 .
\)
Thus condition~\eqref{condition_ISS} is equivalent to
\(
q_3^{\max}Q
\left[
(\gamma-1)m_2^2
+
\frac{\gamma^2}{N-1}
\right]
<
q_3^{\min}Q .
\)
Since $Q>0$, this is equivalent to
\(
(\gamma-1)m_2^2
+
\frac{\gamma^2}{N-1}
<
\frac{q_3^{\min}}{q_3^{\max}}.
\)
Substituting
\(
m_2^2=\gamma-1+\frac{\gamma^2}{N-1},
\)
we obtain
\(
(\gamma-1)m_2^2
+
\frac{\gamma^2}{N-1}
=
(\gamma-1)^2
+
\frac{\gamma^3}{N-1}.
\)
Therefore, it is sufficient to prove
\(
(\gamma-1)^2
+
\frac{\gamma^3}{N-1}
<
\frac{q_3^{\min}}{q_3^{\max}}.
\)
Because $\gamma\ge 1$ and $c_{\mathrm{MPC}}=2c_3$, one has
\(
q_2
=
\max\left\{
\frac{4\gamma c_{\mathrm{MPC}}}{c_3},
\frac{1}{2\gamma},
2\gamma
\right\}
=
8\gamma .
\)
Consequently,
\(
q_3^{\min}
=
\min\left\{
1+\frac{1}{4\gamma},
q_2,
\frac{1}{2\gamma},
1
\right\}
=
\frac{1}{2\gamma},
\)
and
\(
q_3^{\max}
=
\max\left\{
1+\frac{1}{4\gamma},
q_2,
\frac{1}{2\gamma},
1
\right\}
=
8\gamma .
\)
Further using $\gamma\le 1+\xi$, we have
\(
\frac{q_3^{\min}}{q_3^{\max}}
=
\frac{1}{16\gamma^2}
\ge
\frac{1}{16(1+\xi)^2}.
\)
Also,
\(
(\gamma-1)^2
+
\frac{\gamma^3}{N-1}
\le
\xi^2
+
\frac{(1+\xi)^3}{N-1}.
\)
By~\eqref{eq:N-lb-new},
\(
\frac{(1+\xi)^3}{N-1}
<
\frac{1}{16(1+\xi)^2}-\xi^2.
\)
Therefore,
\(
(\gamma-1)^2
+
\frac{\gamma^3}{N-1}
<
\xi^2
+
\frac{1}{16(1+\xi)^2}
-
\xi^2
=
\frac{1}{16(1+\xi)^2}
\le
\frac{q_3^{\min}}{q_3^{\max}}.
\)
Thus,
\(
(\gamma-1)^2
+
\frac{\gamma^3}{N-1}
<
\frac{q_3^{\min}}{q_3^{\max}},
\)
which implies
\(
q_3^{\max}\left( Q\gamma-Q\right)m_2^2
-
q_3^{\min}Q
+
\frac{q_3^{\max}Q\gamma^2}{N-1}
<0.
\)
Hence condition~\eqref{condition_ISS} holds.

% Define, for brevity, $\rho \;\coloneqq\; (1+\bar{E}_w)^2\ge 1$. With this notation, we have $m_2^2 \;=\;\frac{\gamma^2}{N-1}+\gamma-1.$ For \eqref{condition_ISS}, after substitution, it suffices to guarantee
% \(
% \begin{aligned}
%     &(\gamma\rho-1)\Big(\tfrac{\gamma^2}{N-1}+\gamma-1\Big)-1+\tfrac{\gamma^2}{N-1}<0.
% \end{aligned}
% \)
% Since $\gamma>0$, this is equivalent to
% \(
% \frac{\gamma^2}{N-1}+\gamma \;<\; 1+\frac{1}{\rho}
% \;=\;1+\frac{1}{(1+\bar{E}_w)^2}.
% \)
% If $\gamma<1+\frac{1}{\rho}$, it suffices that
% \(
% N \ge 1+\tfrac{\gamma^2}{\bigl(1+\tfrac{1}{\rho}\bigr)-\gamma}.
% \) Because
% \(
% \gamma=\tfrac{Q+\lambda_{\max}(R)\| \mathcal{R}_*^\dagger A^\kappa \|^2}{Q(1-p^2)}
% = \tfrac{1+\frac{\lambda_{\max}(R)\| \mathcal{R}_*^\dagger A^\kappa \|^2}{Q}}{1-p^2},
% \)
% the choice \eqref{eq:R-sufficient} ensures
% \(
% 1+\frac{\lambda_{\max}(R)\| \mathcal{R}_*^\dagger A^\kappa \|^2}{Q}\ \le\ 1+\kappa.
% \)
% Hence $\gamma \le \tfrac{1+\kappa}{1-p^2}$. To make $\gamma<1+\frac{1}{\rho}=1+\frac{1}{(1+\bar{E}_w)^2}$, it suffices that
% \(
% \tfrac{1+\kappa}{1-p^2} < 1+\frac{1}{(1+\bar{E}_w)^2}
% \,\Longleftrightarrow \,
% p^2< \tfrac{\frac{1}{(1+\bar{E}_w)^2}-\kappa}{1+\frac{1}{(1+\bar{E}_w)^2}}.
% \)
% Using $p^2=e^{-2\,\sigma_{\min}(\mathcal{R}_*)(u_{\max}-C)}$ yields exactly the lower bound
% \eqref{eq:umax-sufficient} on $u_{\max}$.

\section*{Proof of Lemma \ref{localxlipschitzMPC}}

We first show the following supporting lemma.
\begin{lemma}
    Consider system \eqref{sys1} under Assumption \ref{assump2}, for each system matrices $A,B$ and $N \in \mathbb{N}, N \geq \kappa$, there exists $\bar{\epsilon}_1>0$ such that for all $\hat{\theta} = [\hat{A} \ \hat{B}] \in \mathcal{B}_{\bar{\epsilon}_1}(\theta_*)$, the minimal singular value $\sigma_{min}(\mathcal{R}_N(\hat{A},\hat{B})) > 0$, where $\mathcal{R}_N(\hat{A},\hat{B}) = \begin{bmatrix}
     \hat{A}^{N-1} \hat{B} &  \hat{A}^{N-2} \hat{B} & \dots & \hat{B} \end{bmatrix} $.
    \label{perturbationReachabilityMatrix}
\end{lemma}
\begin{proof}
    Denote $\Delta_A,\Delta_B$ as perturbations of $A,B$ respectively so that $\hat{A} = A + \Delta_A, \hat{B} = B + \Delta_B$. According to Weyl's Inequality,
\(
\begin{aligned}
    |\sigma_{\min}(\mathcal{R}_N(A,B))-\sigma_{\min}(\mathcal{R}_N(\hat{A},\hat{B}))| \le \| [A\ B] - [\hat{A} \ \hat{B}] \|.
\end{aligned}
\)
Hence, if the perturbation is small enough such that
\(
\| [A\ B] - [\hat{A} \ \hat{B}] \| < \sigma_{\min}(\mathcal{R}_N(A,B)),
\)
then it follows that $\sigma_{\min}(\mathcal{R}_N(\hat{A},\hat{B})) > 0$. 

Denote the difference between two reachability matrices 
    \begin{equation}
    \begin{aligned}
        \Delta_N &= \mathcal{R}_N(\hat{A},\hat{B})-\mathcal{R}_N(A,B) \\
        &= \left[ (\hat{B}-B) \quad ... \quad (\hat{A}^{N-1}\hat{B} - A^{N-1}B) \right],
    \end{aligned}
    \end{equation}
    which can be bounded as below
    \[
    \begin{aligned}
    & \lVert \Delta_N \rVert \leq  \sum_{k=1}^N \|\hat{A}^k\hat{B} - A^kB \| \\
    & = \sum_{k=1}^N \|(A+\Delta_A)^{k-1}(B+\Delta_B)-A^{k-1}B \| \\
    & \leq \sum_{h=1}^{N}  \sum_{j=1}^{h} \tfrac{h!}{(h-j)!j!}\lVert A^{h-j} \rVert \Delta_A^j(\lVert B \rVert + \Delta_B)+ \sum_{h=0}^{i} \lVert A^h \rVert \Delta_B 
    \end{aligned}
    \]
    where the last inequality follows from triangle inequality of norms. Inspecting the RHS of the above inequality, we can conclude that the upper bounds of $\lVert \mathcal{R}_N(A,B)-\mathcal{R}_N(\hat{A},\hat{B}) \rVert$ decrease as the perturbation bounds $\Delta_A,\Delta_B$ decrease. Therefore, for each $A,B,N$, there must exist an upper bound $\bar{\epsilon}_1>0$ such that if $\max \{\Delta_A,\Delta_B \} \leq \bar{\epsilon}_1$, we have $\sum_{h=1}^{N}  \sum_{j=1}^{h} \frac{h!}{(h-j)!j!}\lVert A^{h-j} \rVert \Delta_A^j(\lVert B \rVert + \Delta_B)+ \sum_{h=0}^{i} \lVert A^h \rVert \Delta_B  \leq \sigma_{\min}(\mathcal{R}_N(A,B))$, which implies $\sigma_{\min}(\mathcal{R}_N(\hat{A},\hat{B}))>0$.  
\end{proof}

Now we are ready to state the proof of Lemma \ref{localxlipschitzMPC}.
Denote $U_k$ a concatenation of $u_1$ to $u_k$, and $\hat{\theta}=[\hat{A} \ \hat{B}]$. We define $k$ step system evolution in a compact form as 
\begin{equation}
\begin{split}
x_k &= \hat{A}^k x_0 + \begin{bmatrix}
\hat{A}^{k-1} \hat{B} &  \hat{A}^{k-2} \hat{B} & \dots & \hat{B} \end{bmatrix} U_k\\
& = \hat{A}^k x_0 + \mathcal{R}_k(\hat{A},\hat{B})U_k 
\end{split}
\label{eq:x_kexpanded}
\end{equation}
and so the cost function can be reformulated as 
\begin{equation}
\label{eq:costexpanded}
\begin{aligned}
f(U_k,x_0,\hat{A},\hat{B}) =
\sum_{i=0}^{N-1} \left[
Q\sigma(x_k) + u_k^\top Ru_k
\right] + Q\sigma(x_N)
\end{aligned}
\end{equation}
subject to input constraints
\begin{equation}
\begin{bmatrix}
I_{N\kappa m} \\
- I_{N\kappa m}
\end{bmatrix}
U_k
+
\begin{bmatrix}
- u_{\max}\mathbf{1}_{N\kappa m} \\
\;\;u_{\max}\mathbf{1}_{N\kappa m}
\end{bmatrix}
\le \mathbf{0}.
\label{eq:inputsconstraintsmatrix}
\end{equation}
Define matrix
\(
M
=
\begin{bmatrix}
\nabla_1 H_1(u)
\end{bmatrix}
=
\begin{bmatrix}
I_{N\kappa m} \\
- I_{N\kappa m}
\end{bmatrix}.
\)
We are now verifying the following four conditions for any given fixed $x_0 \in \mathbb{X}_l$ and consider $\hat{A},\hat{B}$ as the parameters.
\begin{enumerate}
\item Since $H_1$ is a constant matrix, so it is $\mathbf{C}^2$, and the components of $H_1(\cdot)$ are convex for any $\hat{A},\hat{B}$ within their domains. We now prove that $f(U_k,x_0,\hat{A},\hat{B})$ is also $\mathbf{C}^2$ by showing that it is a composition of two $\mathbf{C}^2$ functions. Firstly, $\sigma(x)$ is a $\mathbf{C}^2$ function wrt $x$. It suffices to prove that $x_k$ is a $\mathbf{C}^2$ function of $\hat{A},\hat{B},U_k$. From (\ref{eq:x_kexpanded}), $\hat{A}^k x_0$ is a polynomial function of $\hat{A}$, and is therefore a $C^\infty$ function of $\hat{A}$. Similarly, $\mathcal{R}_k(\hat{A},\hat{B})$ is a polynomial function of $\hat{A},\hat{B}$ and $\mathcal{R}_k(\hat{A},\hat{B})U_k$ is a linear operator of $U_k$, which are both $C^\infty$ functions. Therefore, $x_k$ is derived from compositions of many functions that are at least $\mathbf{C}^2$ continuous. Therefore, $f(U_k,x_0,\hat{A},\hat{B})$ is a $\mathbf{C}^2$ continuous function with respect to all its arguments.

\item For any $\hat{A},\hat{B}$, there exists a unique solution to the problem because it is a convex problem with convex constraints.
\item Both $\hat{A},\hat{B}$ and $U_k$ are bounded in norm, and $|x_0| \leq \bar{x}_s$ when MPC is activated.
\item $M^\intercal$ is full row rank, so there always exists a constant $\beta = \sigma_{min}(M^\intercal) > 0$ such that $\|M^\intercal \lambda \| \geq \beta \|\lambda \|$ for any $\lambda$. Besides,
\begin{equation}
\begin{split}
\frac{\partial^2 f}{\partial U_k^2}
&=
\sum_{k=0}^{N-1}
\Bigl\{
\mathcal{R}_k(\hat{A},\hat{B})^\intercal
\Bigl[
2Q\!\left(2e^{2|x_k|} - e^{|x_k|}\right)
\Bigr]\\
& \quad \times \mathcal{R}_k(\hat{A},\hat{B})
+ 
2R
\Bigr\} + 
\mathcal{R}_N(\hat{A},\hat{B})^\intercal \\
& \quad \times 
\Bigl[
2Q\!\left(2e^{2|x_N|} - e^{|x_N|}\right)
\Bigr]
\mathcal{R}_N(\hat{A},\hat{B}).
\end{split}
\end{equation}
From Lemma \ref{perturbationReachabilityMatrix}, for each $A,B,k$, there exists $\bar{\epsilon}_1$ such that for all $[\hat{A} \ \hat{B}] \in \mathcal{B}_{\bar{\epsilon}_1}(\theta_*)$, $\mathcal{R}_k(\hat{A},\hat{B})$ is full rank. Besides $2Q\bigl(2 \exp\bigl(2|x|\bigr) - \exp\bigl(|x|\bigr)\bigr)$ is positive. So $\frac{\partial^2 f}{\partial U_k^2}$ is positive definite, and there exists $\alpha>0$ such that $v^\intercal \frac{\partial^2 f}{\partial U_k^2} v \geq \alpha \|v\|^2$ for any size-conforming vector $v$.
\end{enumerate}

According to Appendix D in \cite{hager1979lipschitz}, there exists a finite Lipschitz constant for $\pi_{\mathrm{MPC}}(x,\theta)$ w.r.t $\theta$. Finally, we note that the Lipschitz constant $D(\bar{x}_s)$ is defined as  
\[
D(\bar{x}_s) := \sup_{\substack{\|x\|\le \bar{x}_s \\ \hat{\theta},\,\theta \in \mathcal{B}_{\bar{\epsilon}_1}(\theta^*)}}
\frac{\|\pi_{\mathrm{MPC}}(x,\hat{\theta}) - \pi_{\mathrm{MPC}}(x,\theta)\|}{\|\hat{\theta} - \theta\|}.
\]
Let $0 < \bar{x}_s^1 \le \bar{x}_s^2$ and define  
$\mathbb{X}_l^1 := \{x : \|x\| \le \bar{x}_s^1\}$,  
$\mathbb{X}_l^2 := \{x : \|x\| \le \bar{x}_s^2\}$.  
Since $\mathbb{X}_l^1 \subseteq \mathbb{X}_l^2$, it follows from the definition of the supremum that  
\(
D(\bar{x}_s^1) = \sup_{x \in \mathbb{X}_l^1} \dots \ \ \le\ \ \sup_{x \in \mathbb{X}_l^2} \dots = D(\bar{x}_s^2).
\)
Therefore, $D(\bar{x}_s)$ is non-decreasing in $\bar{x}_s$.

% \section*{Proof of Lemma \ref{lemma:kappato1step_dragan}}

\section*{Proof of Lemma \ref{lem:levelsetuncertainty}}
The proof requires the following lemma.
\begin{lemma}
For each compact set $\mathbb{X}_l =\{ x \mid |x| \leq \bar{x}_s \}$ with a constant $\bar{x}_s>0$, there exists positive constants $\bar{\epsilon}_1>0$ and $L<\infty$ such that $\forall x \in \mathbb{X}_l, \forall  \hat{\theta} \in \mathcal{B}_{\bar{\epsilon}_1}(\theta^*),$ 
\[
    \big|Y_{\mathrm{MPC}}(\theta^*, x)-Y_{\mathrm{MPC}}(\hat{\theta}, x)\big|  \le L\big(\|\hat{\theta}- \theta^*\|\big).
\]
\label{lem:Y1locallipschitz}
\end{lemma}
\begin{proof}
    From \eqref{eq:x_kexpanded} and \eqref{eq:costexpanded}, we know that the value function $Y_{\mathrm{MPC}}(\hat{\theta},x)$ is at least $\mathbf{C}^1$ continuous. Therefore, according the mean value theorem, within any compact sets of all its arguments, $Y_{\mathrm{MPC}}(\hat{\theta},x)$ is locally Lipschitz continuous \cite[Lemma in p389]{hirsch2013differential}. 
% Namely, for all $x \in \mathbb{X}_l$, with compact constraints on the input \eqref{eq:inputsconstraintsmatrix} and the parameter satisfying $\hat{\theta} \in \mathcal{B}_{\bar{\epsilon}_1}(\theta^*)$, $Y_{\mathrm{MPC}}(\hat{\theta},x)$ is Lipschitz continuous in all its arguments. 
\end{proof}

    Firstly of all, $\frac{1}{2}c_3 < c_4 < c_{\mathrm{MPC}}$ guarantees that $c_{\mathrm{MPC}} - c_3+c_4 > 0$. Since $Y_{\mathrm{MPC}}(\hat{\theta},x) \geq Q \sigma(x)$ and $ Y_{\mathrm{MPC}}(x) \geq Q \sigma(x)$ hold due to the positive definiteness, for ant $|x| \geq \ln\left(\sqrt{\frac{c_{\mathrm{MPC}}}{Q}}+1 \right)$, we have $Y_{\mathrm{MPC}}(\hat{\theta},x) \geq c_{\mathrm{MPC}}, Y_{\mathrm{MPC}}(x) \geq c_{\mathrm{MPC}} \geq c_{\mathrm{MPC}} - c_4$. For any $|x| < \ln\left(\sqrt{\frac{c_{\mathrm{MPC}}}{Q}}+1 \right)$, using Lemma \ref{lem:Y1locallipschitz}, there must exist a small enough $\bar{\epsilon}_3 \leq \bar{\epsilon}_1$ such that if $\|\hat{\theta}-\theta^*\| \leq  \bar{\epsilon}_3$, we have
    \(
    |Y_{\mathrm{MPC}}(\hat{\theta},x)- Y_{\mathrm{MPC}}(x)| \leq L \, \bar{\epsilon}_3 \leq c_4,
    \)
    so that 
    \(
    Y_{\mathrm{MPC}}(\hat{\theta},x) \geq c_{\mathrm{MPC}} \;\Rightarrow\; Y_{\mathrm{MPC}}(x) \geq c_{\mathrm{MPC}} - c_4.
    \)
    The proof of 
    \(
    Y_{\mathrm{MPC}}(\hat{\theta},x) \leq c_{\mathrm{MPC}}-c_3 \;\Rightarrow\; Y_{\mathrm{MPC}}(x) \leq c_{\mathrm{MPC}} - c_3+c_4
    \)
    is similar.

\end{document}